\documentclass[12pt,a4paper]{article}

\usepackage{fullpage}
\usepackage[utf8]{inputenc}
\usepackage{amssymb}
\usepackage{amsmath}
\usepackage{amsthm}
\usepackage[round]{natbib}
\usepackage{url}
\usepackage{graphicx}
\usepackage{multirow}
\usepackage{setspace}
\usepackage{enumerate}
\usepackage{color}
\usepackage{booktabs,rotating}
\usepackage{bm}
\usepackage[nolists]{endfloat}
\usepackage{endrotfloat}
\usepackage{xcolor,colortbl}

\definecolor{Gray}{gray}{0.9}
\newcolumntype{a}{>{\columncolor{Gray}}c}
\newcolumntype{b}{>{\columncolor{white}}c}

\newcommand{\eps}{\varepsilon}

\newcommand{\R}{\mathbb{R}}
\newcommand{\Rbar}{\overline{\R}}
\newcommand{\dd}{\mathrm{d}}

\newcommand{\DD}{\mathcal{D}}

\newcommand{\FF}{\mathcal{F}}

\newcommand{\Mb}{\mathbf{M}}

\newcommand{\D}{\mathbb{D}}
\newcommand{\B}{\mathbb{B}}

\newcommand{\Pb}{\mathbb{P}}
\newcommand{\M}{\mathbb{M}}
\newcommand{\U}{\mathbb{U}}

\newcommand{\W}{\mathbb{W}}
\newcommand{\X}{\mathbb{X}}
\newcommand{\Y}{\mathbb{Y}}
\newcommand{\Z}{\mathbb{Z}}

\newcommand{\G}{\mathbb{G}}
\newcommand{\Ex}{\mathrm{E}}
\newcommand{\Var}{\mathrm{Var}}
\newcommand{\Cov}{\mathrm{Cov}}

\newcommand{\1}{\mathbf{1}}

\newcommand{\ip}[1]{\lfloor #1 \rfloor}

\renewcommand{\Pr}{\mathrm{P}}

\newcommand{\as}{\overset{\mathrm{a.s.}}{\longrightarrow}}

\theoremstyle{plain}
\newtheorem{prop}{Proposition}[section]

\newtheorem{cond}[prop]{Condition}

\numberwithin{equation}{section}

\def\disp{\displaystyle}

\title{Detecting distributional changes in samples of independent block maxima using probability weighted moments}

\author{
  Ivan Kojadinovic\,\footnote{Universit\'e de Pau et des Pays de l'Adour, 
Laboratoire de math\'ematiques et applications, 
UMR CNRS 5142, B.P. 1155, 64013 Pau Cedex, France.
{E-mail:} \texttt{ivan.kojadinovic@univ-pau.fr}}
\and
  Philippe Naveau\,\footnote{Laboratoire des sciences du climat et de l'environnement, CNRS, Orme des Merisiers / B\^at.\ 701 C.E.A Saclay, 91191 Gif-sur-Yvette, France.
{E-mail:} \texttt{philippe.naveau@lsce.ipsl.fr}}
}

\begin{document}
\maketitle

\begin{abstract}
The analysis of seasonal or annual block maxima is of interest in fields such as hydrology, climatology or meteorology. In connection with the celebrated method of block maxima, we study several tests that can be used to assess whether the available series of maxima is identically distributed. It is assumed that block maxima are independent but not necessarily generalized extreme value distributed. The asymptotic null distributions of the test statistics are investigated and the practical computation of approximate p-values is addressed. Extensive Monte-Carlo simulations show the adequate finite-sample behavior of the studied tests for a large number of realistic data generating scenarios. Illustrations on several environmental datasets conclude the work.

\medskip


\medskip

\end{abstract}


\section{Introduction}

The {\em block maximum method} put forward in the seminal monograph of \cite{Gum58} is frequently applied in environmental sciences to analyse extremes of a given series of observations. It consists of dividing the initial high frequency observations into blocks of equal size and forming a sample of block maxima. A typical example is the study of annual or seasonal maxima of temperatures or precipitation from daily observations. The distribution of the block maxima is classically modeled using the generalized extreme value (GEV) distribution. Under rather broad conditions on the underlying initial observations, the extremal types theorem \citep[see, e.g.,][]{DehFer06} states that, as the block size increases to infinity, the only possible non-degenerated limit of affinely normalized block maxima is the GEV. The cumulative distribution function (c.d.f.) of this three-parameter distribution is
\begin{equation}
\label{eq:gev}
G_{\mu,\sigma,\xi}(x) =
\left\{
\begin{array}{ll}
\disp \exp \left\{ - \left(1 + \xi \cdot \frac{x - \mu}{\sigma} \right)_+^{-1/\xi} \right\}, & \mbox{if } \xi \neq 0, \\
\\
\disp \exp \left\{ - \exp \left( - \frac{x - \mu}{\sigma} \right) \right\}, & \mbox{if } \xi=0,
\end{array}
\right.
\end{equation}
for $x \in \R$, where $\mu \in \R$, $\sigma > 0$ and $\xi \in \R$ are the so-called location, scale and shape parameters, respectively, and with the notation $(\cdot)_+ = \max(\cdot,0)$. 

The adequacy of the GEV distribution as a model for a sample of block maxima is predicted by the theory when the underlying high frequency observations are identically distributed and independent, or satisfy some suitable mixing conditions \citep[see, e.g.,][]{LeaLinRoo83}, and the block size is ``sufficiently large''. Sometimes maxima are directly observed, making it impossible to question the stationarity of the initial series. Furthermore, the block maxima method may be applied even when the initial series is clearly not identically distributed. For instance, when the original observations exhibit seasonality, it suffices to take the block size to be a multiple of the season length.

When the block size is large, one of the least questionable assumptions might be that of independence among block maxima. Before fitting a GEV distribution to the resulting extremes, two questions then arise naturally: (i) Is the series of independent maxima identically distributed? (ii) Is the block size large enough to warrant a good fit by the GEV distribution? The aim of this work is to focus on the first question without assuming that the answer to the second question is positive. Specifically, we consider that we have at hand a sample of independent block maxima $X_1,\dots,X_n$ and our objective is to develop tests for
\begin{equation}
\label{eq:H0}
  H_0 : \,\exists \, F \text{ such that } X_1, \ldots, X_n \text{ have c.d.f.\ } F
\end{equation}
against alternatives involving the non-constancy of the c.d.f. In other words, we aim at developing tests for detecting changes in the distribution of independent block maxima.

Under the additional assumption that each $X_i$ is GEV distributed, a natural approach for testing $H_0$ is to apply the maximum likelihood based theory developed in the first chapter of the well-known monograph of \cite{CsoHor97} on change-point analysis. The latter approach was considered for instance in \cite{JarRen08} but not without significant practical difficulties related to maximum likelihood estimation of the GEV parameters from small samples \citep[see also for instance the discussions in][]{HosWalWoo85,DieGuiNavRib08}.

The approach that we consider is based on {\em probability weighted moments} (PWM) \citep[see, e.g.,][]{GreLanWal79,HosWalWoo85,HosWal87,DieGuiNavRib08,GuiNavDieRib09,FerdeH15} and has the advantage of giving rise to tests that can hold their level asymptotically even if the distributions of the block maxima $X_i$ are not GEV. To illustrate that point, in our simulations, some scenarios involved block maxima obtained from blocks of size 1, 5, 10 and 50. The later sentence actually implies that the tests can be applied outside of the block maxima setting, although, by construction, this is less natural. These aspects will be clarified in the forthcoming sections.

The rest of the paper is organized as follows. In Section~\ref{sec:PWM}, we recall the PWM method for estimating the parameters of the GEV due to \cite{HosWalWoo85}, propose related {\em cumulative sum} (CUSUM) statistics for change-point detection, establish the limiting null distributions of the latter and address the computation of approximate p-values. The third section proposes alternative tests in the framework of the generalized PWM method for estimating the GEV due to \cite{DieGuiNavRib08}. Simulation results are partially reported in the fourth section and show the adequate behavior of the tests for many realistic data generating scenarios. The last section presents several fully reproducible illustrations on environmental datasets. 

All proofs are deferred to a sequence of appendices. Appendices~\ref{sec:weighted_uniform} and~\ref{proof:prop:weak_Mn} in particular establish necessary intermediate results that might be of independent interest, namely the weak convergence of the sequential weighted uniform empirical process and of certain sequential empirical processes constructed from PWM estimators. The paper is also accompanied by supplementary material available upon request from the authors and mostly consisting of simulation results that are not reported in Section~\ref{sec:sims}. Finally, note that the studied tests for change-point detection are implemented in the package {\tt npcp} \citep{npcp} for the \textsf{R} statistical environment \citep{Rsystem}.

\section{The PWM method and related statistics}
\label{sec:PWM}

The statistics for testing $H_0$ in~\eqref{eq:H0} that we propose are related to the method of PWM for estimating the GEV parameters \citep{HosWalWoo85}. After briefly recalling the method and suggesting related statistics for change-point detection, we study the limiting null distributions of the latter and address the practical computation of approximate p-values.

\subsection{Estimation of the GEV using PWM}
\label{sec:PWM_method}

A general definition of PWM can be found for instance in \citet[Section~2]{HosWalWoo85} (see also Appendix~\ref{proof:prop:weak_Mn}). Let $X$ be a random variable with c.d.f.\ $F$ and consider the functions $\nu_1 (x) = 1$, $\nu_2(x) = x$ and $\nu_3(x) = x^2$, $x \in [0,1]$. The quantities $\beta_i = \Ex [ X \nu_i\{F(X)\} ]$, $i \in \{1,2,3\}$, are the three (probability weighted) moments used in the PWM method for estimating the GEV parameters. 

Assume that $H_0$ in~\eqref{eq:H0} holds with $F$ the c.d.f.\ of the GEV given in~\eqref{eq:gev}. The method exploits the following relationship between $\beta_{1}$, $\beta_{2}$ and $\beta_{3}$ and the parameters $\mu$, $\sigma$ and $\xi$ of the GEV. For $\sigma > 0$, $\xi < 1$ and $\xi \neq 0$, there holds
\begin{equation}
\label{eq:sys_PWM}
\left\{
\begin{array}{l}
\disp \beta_{1} = \mu - \frac{\sigma}{\xi} \{ 1 - \Gamma(1 - \xi) \}, \\
\disp 2 \beta_{2} -  \beta_{1} = \frac{\sigma}{\xi} \Gamma(1 - \xi) (2^\xi - 1), \\
\disp \frac{3  \beta_{3} -  \beta_{1}}{2  \beta_{2} -  \beta_{1}} = \frac{3^\xi - 1}{2^\xi - 1},
\end{array}
\right.
\end{equation}
where $\Gamma$ is the Gamma function. Using continuity arguments, some thought reveals that the constraint $\xi \neq 0$ is actually unnecessary. Estimators of $\mu$, $\sigma$ and $\xi$ are then the solution of the above system of equations in which $\beta_{1}$, $\beta_{2}$ and $\beta_{3}$ are replaced by their estimators.

Two ways to estimate $\beta_{1}$, $\beta_{2}$ and $\beta_{3}$ were considered by \cite{HosWalWoo85}. The first approach consists of first estimating $F$ by a slight modification of the empirical c.d.f.\ of the sample, and then estimating the PWM by suitable arithmetic means. For reasons that shall become clear later in this section, we present the resulting estimators when based on a subsample $X_k,\dots,X_l$, $1 \leq k \leq l \leq n$, of the available data. The c.d.f.\ $F$ is estimated by
\begin{equation}
\label{eq:Fkl}
F_{k:l}(x) = \frac{1}{l-k+1} \left\{ \sum_{j=k}^l \1(X_j \leq x) + \gamma \right\}, \qquad x \in \R,
\end{equation}
with the convention that $F_{k:l} = 0$ for $k > l$ and where $\gamma$ is a constant whose value suggested by \cite{HosWalWoo85} is -0.35. A sensible estimator of $\beta_{i}$ is then
\begin{equation}
\label{eq:hatbetakl}
\hat \beta_{i,k:l} = \frac{1}{l-k+1}  \sum_{j=k}^l X_j \nu_i\{F_{k:l}(X_j)\}
\end{equation}
with the convention that $\hat \beta_{i,k:l} = 0$ for $k > l$. 

A second equally natural possibility to estimate $\beta_{1}$, $\beta_{2}$ and $\beta_{3}$ is to consider the asymptotically equivalent estimators $\hat b_{1,1:n},\hat b_{2,1:n},\hat b_{3,1:n}$ proposed by \cite{LanMatWal79} and defined, for a subsample $X_k,\dots,X_l$, $1 \leq k \leq l \leq n$, by
\begin{equation}
\label{eq:landwehr}
\hat b_{i,k:l} = \frac{1}{l-k+1} \sum_{j=k}^l \frac{\prod_{m=1}^{i-1}\{ (j-k+1)- m\}}{\prod_{m=1}^{i-1} \{ (l-k+1)-m \}} X_{(j),k:l},
\end{equation}
where $X_{(k),k:l} \leq \dots \leq X_{(l),k:l}$ are the order statistics computed from $X_k,\dots,X_l$, and with the convention that $\hat b_{i,k:l} = 0$ if $l - k < i$. The latter estimators are unbiased (they are related to $U$-statistics). 

\citet[Section~4]{HosWalWoo85} claimed that the $\hat \beta_{i,k:l}$ are slightly superior to the $\hat b_{i,k:l}$ in terms of finite-sample performance. While this seems true when $X_1,\dots,X_n$ are i.i.d.\ from a GEV with location zero and scale one, it does not seem to necessarily hold for arbitrary location and scale parameters.

Let $g_\mu$, $g_\sigma$ and $g_\xi$ be the components of the map from $\R^3$ to $\R^3$ implicitly defined by~\eqref{eq:sys_PWM} that transforms $\bm \beta = (\beta_{1},\beta_{2},\beta_{3})$ into $(\mu,\sigma,\xi)$ (that is, $\mu = g_\mu(\bm \beta)$, $\sigma = g_\sigma(\bm \beta)$ and $\xi = g_\xi(\bm \beta)$), and let us first focus on the function $g_\xi$  without assuming anymore that $F$ is the c.d.f.\ of the GEV but only that $F$ is continuous and that $\beta_{1}$, $\beta_{2}$ and $\beta_{3}$ exist. Under the constraint $\xi < 1$, it can be easily verified that the solution $g_\xi(\bm \beta)$ of the third equation in~\eqref{eq:sys_PWM} exists and is unique if and only if $1 < (3  \beta_{3} -  \beta_{1})/(2  \beta_{2} -  \beta_{1}) < 2$. The latter is guaranteed by the following result proved in Appendix~\ref{proof:prop:beta_ineqs}.

\begin{prop}
\label{prop:beta_ineqs}
Assume that $F$ is continuous and that $\beta_{1}$, $\beta_{2}$ and $\beta_{3}$ exist. Then, $2 \beta_{2} - \beta_{1} > 0$, $3 \beta_3 - 2 \beta_2 > 0$ and $- \beta_{1} + 4 \beta_{2} - 3 \beta_{3} > 0$.
\end{prop}

Given the equivalence
$$
\left\{
\begin{array}{l}
2 \beta_{2} - \beta_{1} > 0, \\
3 \beta_3 - 2 \beta_2 > 0, \\
- \beta_{1} + 4 \beta_{2} - 3 \beta_{3} > 0,
\end{array}
\right.
\iff
\left\{
\begin{array}{l}
2 \beta_{2} - \beta_{1} > 0, \\
3 \beta_{3} - \beta_{1} > 0, \\ 
1 < (3  \beta_{3} -  \beta_{1})/(2  \beta_{2} -  \beta_{1}) < 2,
\end{array}
\right.
$$
it is natural to consider $g_\xi$ as a function defined on the convex subset
\begin{equation}
\label{eq:def_xi}
\DD_\xi = \{ \bm x \in \R^3 : 2 x_2 - x_1 > 0, 3 x_3 - 2x_2 > 0, - x_1 + 4x_2 - 3x_3 > 0 \}.
\end{equation}
In addition, it is easy to see that, since $\bm \beta \in \DD_\xi$, necessarily $g_\xi(\bm \beta) < 1$ and $g_\sigma(\bm \beta) > 0$, which is fully in accordance with the constraints under which the system~\eqref{eq:sys_PWM} was derived. Using the fact that $x \mapsto x/(2^x-1)$ can be extended by continuity at 0 and is strictly positive, it is also natural to consider $g_\sigma$ and $g_\mu$ as functions defined on $\DD_\xi$.

When $\beta_1$, $\beta_2$ and $\beta_3$ are estimated using the unbiased estimators $\hat b_{1,1:n},\hat b_{2,1:n},\hat b_{3,1:n}$ defined analogously to~\eqref{eq:landwehr}, \citet[Appendix B]{HosWalWoo85} showed that $\hat{\bm b}_{1:n} = (\hat b_{1,1:n}, \hat b_{2,1:n},\hat b_{3,1:n}) \in \DD_\xi$ for all $n \geq 3$, where $\DD_\xi$ is defined in~\eqref{eq:def_xi}, which implies that $g_\xi(\hat{\bm b}_{1:n}) < 1$ and $g_\sigma(\hat{\bm b}_{1:n}) > 0$ for all $n \geq 3$. The latter desirable property is referred to as the {\em feasibility criterion} of the PWM method. When $\hat{\bm \beta}_{1:n} = (\hat \beta_{1,1:n}, \hat \beta_{2,1:n},\hat \beta_{3,1:n})$ is used instead of $\hat{\bm b}_{1:n}$, an analogue result is not available. However, provided $\Ex(|X_1|) < \infty$, a consequence of the strong law of large numbers and the Glivenko--Cantelli lemma is that $\hat{\bm \beta}_{1:n}$ converges almost surely to $\bm \beta$ (see Appendix~\ref{proof:prop:weak_S_g_n}). The latter immediately implies that $\Pr(\hat{\bm \beta}_{1:n} \in \DD_\xi) \to 1$.

\subsection{Test statistics}

The above developments suggest to study tests for change-point detection based on the three CUSUM like statistics
\begin{equation}
\label{eq:S_g_n}
S_{g,n} =  \max_{1 \leq k \leq n-1} \frac{k (n-k)}{n^{3/2}} \1( \hat{\bm \beta}_{1:k} \in \DD_\xi, \hat{\bm \beta}_{k+1:n} \in \DD_\xi) \left| g( \hat{\bm \beta}_{1:k}) - g(\hat{\bm \beta}_{k+1:n}) \right|, \, g \in \{g_\mu, g_\sigma, g_\xi\},
\end{equation}
or on the three alternative ones
\begin{equation}
\label{eq:T_g_n}
T_{g,n} =  \max_{3 \leq k \leq n-3} \frac{k (n-k)}{n^{3/2}} \left| g( \hat{\bm b}_{1:k}) - g(\hat{\bm b}_{k+1:n}) \right|, \qquad g \in \{g_\mu, g_\sigma, g_\xi\}.
\end{equation}
In~\eqref{eq:S_g_n}, $\hat{\bm \beta}_{1:k}=(\hat \beta_{1,1:k},\hat \beta_{2,1:k},\hat \beta_{3,1:k})$ and $\hat{\bm \beta}_{k+1:n}=(\hat \beta_{1,k+1:n},\hat \beta_{2,k+1:n},\hat \beta_{3,k+1:n})$ are defined according to~\eqref{eq:hatbetakl}, and $\DD_\xi$ is given in~\eqref{eq:def_xi}, while in~\eqref{eq:T_g_n},  $\hat{\bm b}_{1:k}=(\hat b_{1,1:k},\hat b_{2,1:k},\hat b_{3,1:k})$ and $\hat{\bm b}_{k+1:n}=(\hat b_{1,k+1:n},\hat b_{2,k+1:n},\hat b_{3,k+1:n})$ are defined according to~\eqref{eq:landwehr}.

For instance, under $H_0$ in~\eqref{eq:H0} with $F$ the c.d.f.\ of the GEV, $S_{g_\xi,n}$ is the maximum over $k \in \{1,\dots,n-1\}$ of the normalized absolute difference between $g_\xi(\hat{\bm \beta}_{1:k})$, the estimator of the shape parameter $\xi$ computed from $X_1,\dots,X_k$, and $g_\xi(\hat{\bm \beta}_{k+1:n})$, the same estimator computed from $X_{k+1},\dots,X_n$. The role of the indicator $\1( \hat{\bm \beta}_{1:k} \in \DD_\xi, \hat{\bm \beta}_{k+1:n} \in \DD_\xi)$ is to allow the previous evaluations of $g_\xi$ only if $\hat{\bm \beta}_{1:k} \in \DD_\xi$ and $\hat{\bm \beta}_{k+1:n} \in \DD_\xi$. The coefficient $k (n-k)/n^{3/2}$ in front of each absolute difference in~\eqref{eq:S_g_n} is the classical normalizing term in the CUSUM approach ensuring that, under additional suitable assumptions, $S_{g_\xi,n}$ converges in distribution under the null. The interpretation of $T_{g_\xi,n}$ is similar. The reason for the disappearance of the indicators in its expression is that, as discussed in Section~\ref{sec:PWM_method}, $\hat{\bm b}_{1:k} \in \DD_\xi$ for all $k \geq 3$ \citep[Appendix B]{HosWalWoo85}.

If each block maxima in the sample $X_1,\dots,X_n$ is GEV distributed, $S_{g_\xi,n}$ and $T_{g_\xi,n}$ (resp.\ $S_{g_\sigma,n}$ and $T_{g_\sigma,n}$, $S_{g_\mu,n}$ and $T_{g_\mu,n}$) are thus test statistics particularly sensitive to changes in the shape parameter $\xi$ (resp.\ the scale parameter $\sigma$, the location parameter~$\mu$) of the GEV. It is however important to keep in mind that the proposed tests are actually nonparametric in the sense that they remain potentially meaningful even if the block maxima are not GEV distributed. The latter follows from the fact $\hat \beta_{1,1:n} = \hat b_{1,1:n}$ is simply the sample mean, while, for instance from \citet[page 252]{HosWalWoo85}, we have that $2 \hat b_{2,1:n} - \hat b_{1,1:n}$ is actually Gini's mean difference and, under~$H_0$ in~\eqref{eq:H0}, is thus an estimator of scale as measured by $2\beta_2 - \beta_1 = \Ex(|X_1-X_2|)/2$. Consequently, the tests based on $T_{g_\mu,n}$ and $S_{g_\mu,n}$ (resp.\ $T_{g_\sigma,n}$ and $S_{g_\sigma,n}$) can in general be interpreted as tests for change-point detection particularly sensitive to changes in the location (resp.\ scale) of the distribution. Similarly, the tests based on $T_{g_\xi,n}$ and $S_{g_\xi,n}$ can be interpreted in general as tests for change-point detection particularly sensitive to changes in the upper tail of $F$. Indeed, the ratio $(3 \beta_3 -\beta_1)/(2 \beta_2 -\beta_1)$ appearing in the expression of $g_{\xi}$ can be rewritten as
$$
\frac{\Ex \{ \max(X_1,X_2,X_3) - (X_1+X_2+X_3)/3 \}}{\Ex \{ \max(X_1,X_2)-(X_1+X_2)/2 \}} 
$$
and is thus clearly invariant to the replacement of the $X_i$ by $c X_i + d$ if $c > 0$. Some thought reveals that this ratio captures some information on the upper tail of $F$.

The statement that the tests are actually nonparametric will also become clearer in view of the theoretical results established in the next two subsections and will be illustrated in the Monte Carlo experiments partially reported in Section~\ref{sec:sims}.

\subsection{Limiting null distributions of $S_{g,n}$ and $T_{g,n}$}
\label{sec:limiting}

We shall now state conditions under which $S_{g,n}$ in~\eqref{eq:S_g_n} and $T_{g,n}$ in~\eqref{eq:T_g_n} converge in distribution under the null. 

\begin{cond}
\label{cond:PWM}
The random variables $X_1,\dots,X_n$ are i.i.d.\ from a continuous distribution with c.d.f.\ $F$ such that $\Ex(X_1^2) < \infty$ and there exists $\alpha \in [0,1/2)$ such that $\sup_{x \in \R} | H_\alpha(x) | < \infty$, where $H_\alpha(x) = x [ F(x) \{ 1-F(x)\} ]^\alpha$, $x \in \R$.
\end{cond}

Condition~\ref{cond:PWM} is for instance satisfied when $F$ is the c.d.f.\ of the GEV in~\eqref{eq:gev} with~$\xi < 1/2$, that is, when it has finite variance. In addition, it can be verified that when both $\max(X_1,0)$ and $-\min(X_1,0)$ have an ultimately monotone density and are in the maximum domain of attraction of the GEV with shape parameter $\xi < 1/2$, then  $\sup_{x \in \R} | H_\alpha(x) | < \infty$ for any $\alpha \in (\xi, 1/2)$ (see Section~3 
in the supplementary material).

In the sequel, the first-order partial derivatives of $g \in \{g_\mu, g_\sigma, g_\xi\}$ will be denoted by $\partial_1 g$, $\partial_2 g$ and $\partial_3 g$, respectively. Recall furthermore that $\nu_1(x)=1$, $\nu_2(x) = x$ and $\nu_3(x) = x^2$, $x \in [0,1]$. 

The following result is proved in Appendix~\ref{proof:prop:weak_S_g_n}.

\begin{prop}[Limiting null distributions of $S_{g,n}$ and $T_{g,n}$]
\label{prop:weak_S_g_n}
Under Condition~\ref{cond:PWM}, for any $g \in \{g_\mu, g_\sigma, g_\xi\}$, $S_{g,n} - T_{g,n} =o_\Pr(1)$ and $S_{g,n}$ converges in distribution to $\sigma_{g} \sup_{s \in [0,1]} |\U(s) - s \U(1)|$, where $\U$ is a standard Brownian motion on $[0,1]$ and
\begin{equation}
\label{eq:sigma_g}
\sigma_{g}^2 = \sum_{i,j=1}^3 \partial_i g(\bm \beta) \partial_j g(\bm \beta) \Cov(Y_{\nu_i},Y_{\nu_j}),
\end{equation}
with, for any $\nu \in \{\nu_1,\nu_2,\nu_3\}$,
\begin{equation}
\label{eq:Ynu}
Y_\nu = X_1 \nu\{F(X_1)\} + \int_\R x \nu'\{F(x)\} \1(X_1 \leq x) \dd F(x).
\end{equation}
\end{prop}

\subsection{Computation of approximate p-values}
\label{sec:practical}

Proposition~\ref{prop:weak_S_g_n} suggests to base the computation of approximate p-values for $S_{g,n}$ and $T_{g,n}$, $g \in \{g_\mu, g_\sigma, g_\xi\}$, on their estimated asymptotic null distributions. Indeed, the distribution of the supremum of a Brownian bridge, known as the Kolmogorov distribution, can be approximated very well in practice. It therefore simply remains to estimate the unknown variance $\sigma_{g}^2$ in~\eqref{eq:sigma_g}.

To do so, it is necessary to estimate the first-order partial derivatives of $g$ at $\bm \beta$ and the covariances $\Cov(Y_{\nu_i},Y_{\nu_j})$, $i,j \in \{1,2,3\}$, where $Y_{\nu_i}$ is defined in~\eqref{eq:Ynu}. Natural estimators of the former are $\partial_1 g(\hat{\bm \beta}_{1:n})$, \dots, $\partial_p g(\hat{\bm \beta}_{1:n})$ or $\partial_1 g(\hat{\bm b}_{1:n})$, \dots, $\partial_p g(\hat{\bm b}_{1:n})$, while estimators of the latter can be based on the pseudo-observations 
\begin{equation*}
Y_{\nu,i,n} = X_i \nu\{F_{1:n}(X_i)\} + \frac{1}{n} \sum_{j=1}^n X_j \nu'\{F_{1:n}(X_j)\} \1(X_i \leq X_j),
\end{equation*}
for $i \in \{1,\dots,n\}$ and $\nu \in \{\nu_1,\nu_2,\nu_3\}$, where $F_{1:n}$ is defined in~\eqref{eq:Fkl}. Indeed, a sensible estimator $\widehat \Cov_n(Y_{\nu_i},Y_{\nu_j})$ of $\Cov(Y_{\nu_i},Y_{\nu_j})$, $i,j \in \{1,2,3\}$, under $H_0$ is then simply the sample covariance computed from $(Y_{\nu_i,1,n}, Y_{\nu_j,1,n})$,\dots,$(Y_{\nu_i,n,n}, Y_{\nu_j,n,n})$. 

The proof of the following result is omitted for the sake of brevity.

\begin{prop}
\label{prop:conv_sigma}
Under Condition~\ref{cond:PWM}, for any $g \in \{g_\mu, g_\sigma, g_\xi\}$,
\begin{equation}
\label{eq:hat_sigma_g_n}
\hat \sigma_{g,n}^2 = \sum_{i,j=1}^p \partial_i g(\hat{\bm \beta}_{1:n}) \partial_j g(\hat{\bm \beta}_{1:n}) \widehat \Cov_n(Y_{\nu_i},Y_{\nu_j})
\end{equation}
and
\begin{equation*}
\check \sigma_{g,n}^2 = \sum_{i,j=1}^p \partial_i g(\hat{\bm b}_{1:n}) \partial_j g(\hat{\bm b}_{1:n}) \widehat \Cov_n(Y_{\nu_i},Y_{\nu_j})
\end{equation*}
converge almost surely to $\sigma_{g}^2$ in~\eqref{eq:sigma_g}.
\end{prop}

As we continue, we shall thus compute approximate p-values for $S_{g,n}$ as $1 - F_K(S_{g,n} / \hat \sigma_{g,n} )$ and approximate p-values for $T_{g,n}$ as $1 - F_K(T_{g,n} / \check \sigma_{g,n} )$, where $F_K$ is the c.d.f.\ of the Kolmogorov distribution. As classically done, we approached $F_K$ by the c.d.f.\ of the statistic of the Kolmogorov--Smirnov goodness-of-fit test for a simple hypothesis. From a practical perspective, we used the function {\tt pkolmogorov1x} given in the code of the \textsf{R} function {\tt ks.test}.

An even more important practical issue is the computation of $S_{g,n}$ and $T_{g,n}$ for $g$ equal to $g_\mu$, $g_\sigma$ and $g_\xi$. The numerical solving of the system~\eqref{eq:sys_PWM} can actually be avoided by using the following accurate approximations proposed in \cite{HosWalWoo85}:
\begin{align}
\label{eq:tilde_g_xi}
\xi &= \tilde g_\xi(\bm \beta) = f \circ f_\xi(\bm \beta) \quad \mbox{with} \quad \left\{
\begin{array}{l}
f(x) = - 7.8590 x - 2.9554 x^2, \\
\disp f_\xi(x_1,x_2,x_3) = \frac{2 x_2 - x_1}{3 x_3 - x_1} - \frac{\log 2}{\log 3},
\end{array}
\right. \\
\label{eq:tilde_g_sigma}
\sigma &= \tilde g_\sigma(\bm \beta) = \frac{(2 \beta_{\nu_2} - \beta_{\nu_1}) \xi}{\Gamma(1 - \xi) (2^\xi - 1)}, \\
\label{eq:tilde_g_mu}
\mu &= \tilde g_\mu(\bm \beta) = \beta_{\nu_1}  + \frac{\sigma}{\xi} \{1 - \Gamma(1 - \xi)\}.
\end{align}
We also used the above approximations in combination with symbolic differentiation to obtain approximations of the first-order partial derivatives of $g_\mu$, $g_\sigma$ and $g_\xi$ necessary for computing $\hat \sigma_{g,n}^2$ and $\check \sigma_{g,n}^2$ in Proposition~\ref{prop:conv_sigma}. 

\section{The GPWM method and related statistics}
\label{sec:GPWM}

Alternative test statistics could yet be based on the generalization of the PWM method for estimating the parameters of the GEV proposed by \cite{DieGuiNavRib08}. Under the assumption that $H_0$ in~\eqref{eq:H0} holds with $F$ the c.d.f.\ of the GEV distribution, the estimation method involves the quantities $\beta_i = \Ex [ X \nu_i\{F(X)\} ]$, $i \in \{1,2,3\}$, with $\nu_1(x) = - x \log x$, $\nu_2(x) = x (\log x)^2$ and $\nu_3(x) = - x^2 \log x$, $x \in [0,1]$. The moments $\beta_1$, $\beta_2$ and $\beta_3$ are called {\em generalized} PWM in \cite{DieGuiNavRib08} as they are not PWM in the classical sense considered for instance in \cite{GreLanWal79} because of the definitions of the functions $\nu_1$, $\nu_2$ and $\nu_3$.

The relationship between the GEV parameters $\mu$, $\sigma$ and $\xi$ and the GPWM $\beta_{1}$, $\beta_{2}$ and $\beta_{3}$ is given by the following system of equations, which holds for $\sigma > 0$, $\xi < 2$ and $\xi \neq 0$:
\begin{equation}
\label{eq:sys_GPWM}
\left\{
\begin{array}{l}
\disp 4 \beta_{1} = \mu - \frac{\sigma}{\xi} \{ 1 - 2^\xi \Gamma(2 - \xi) \}, \\
\disp \beta_{1} - \beta_{2} = \frac{\sigma}{2^{3-\xi}} \Gamma(2 - \xi), \\
\disp \frac{2 (\beta_{1} - \beta_{2})}{\beta_{1} - 9/4  \beta_{3}} = \frac{\xi}{1 - (3/2)^\xi}.
\end{array}
\right.
\end{equation}
Again, the constraint $\xi \neq 0$ can be removed using continuity arguments. Estimators of $\mu$, $\sigma$ and $\xi$ are then the solution of the above system in which $\beta_{1}$, $\beta_{2}$ and $\beta_{3}$ are replaced by suitable estimators. A sensible choice is to consider $\hat \beta_{1,1:n}$, $\hat \beta_{2,1:n}$ and $\hat \beta_{3,1:n}$, respectively, defined analogously to~\eqref{eq:hatbetakl}.

Let $h_\mu$, $h_\sigma$ and $h_\xi$ be the components of the map from $\R^3$ to $\R^3$ implicitly defined by~\eqref{eq:sys_GPWM} that transforms $\bm \beta = (\beta_{1}, \beta_{2}, \beta_{3})$ into $(\mu,\sigma,\xi)$. In practice, to evaluate whether $\hat{\bm \beta}_{1:n} = (\hat \beta_{1,1:n}, \hat \beta_{2,1:n}, \hat \beta_{3,1:n}) \in \R^3$ belongs to the domain of definition $\DD_h$ of $h_\xi$, $h_\sigma$ and  $h_\mu$, we shall attempt to solve numerically~\eqref{eq:sys_GPWM} and will conclude that $\hat{\bm \beta}_{1:n} \in \DD_h$ if and only if $h_\xi(\hat{\bm \beta}_{1:n}) < 2$, $h_\sigma(\hat{\bm \beta}_{1:n}) > 0$ and $h_\mu(\hat{\bm \beta}_{1:n}) \in \R$.

The above ingredients lead to the following generalization of the statistics $S_{g,n}$ in~\eqref{eq:S_g_n}:
\begin{equation}
\label{eq:S_h_n}
S_{h,n} =  \max_{1 \leq k \leq n-1} \frac{k (n-k)}{n^{3/2}} \1( \hat{\bm \beta}_{1:k} \in \DD_h, \hat{\bm \beta}_{k+1:n} \in \DD_h) \left| h( \hat{\bm \beta}_{1:k}) - h(\hat{\bm \beta}_{k+1:n}) \right|, \, h \in \{h_\mu, h_\sigma, h_\xi \}.
\end{equation}
Should the block maxima $X_i$ be GEV distributed, $S_{h_\xi,n}$ (resp.\ $S_{h_\sigma,n}$, $S_{h_\mu,n}$) is a test statistic particularly sensitive to changes in the shape (resp.\ scale, location) parameter. 

We conjecture that the asymptotics of the statistics in~\eqref{eq:S_h_n} under $H_0$ in~\eqref{eq:H0} take the same form as those of the statistics $S_{g,n}$ considered in the previous section, the only difference being the definitions of the functions $\nu_1$, $\nu_2$ and~$\nu_3$. Let $\sigma_{h}^2$ be the analogue of~\eqref{eq:sigma_g}. 
Specifically, we conjecture that, under the null, $S_{h,n}$ converges in distribution to $\sigma_{h} \sup_{s \in [0,1]} |\U(s) - s \U(1)|$, where $\U$ is a standard Brownian motion on $[0,1]$. By means of a sequential extension of the method of proof used in \cite{DieGuiNavRib08}, it might be possible to show the aforementioned weak convergence when $F$ is the c.d.f.\ of the GEV and $h$ is either $h_\xi$, $h_\sigma$ or $h_\mu$. 

To carry out the tests, consider the analogue $\hat \sigma_{h,n}^2$ of~\eqref{eq:hat_sigma_g_n}.
Approximate p-values for $S_{h,n}$ can be computed as $1 - F_K(S_{h,n} / \hat \sigma_{h,n} )$, where $F_K$ is the c.d.f.\ of the Kolmogorov distribution.

From a practical perspective, as in the previous section, the numerical solving of system~\eqref{eq:sys_GPWM} can be avoided. Starting from the fact that the function $x \mapsto \{ x/(1-(3/2)^x) \}^{\log(3/2)}$ is almost linear on $[-1,1]$, we propose the following accurate approximations of $h_\xi$, $h_\sigma$ and $h_\mu$:
\begin{align}
\label{eq:tilde_h_xi}
\xi &= \tilde h_\xi(\bm \beta) = f \circ f_\xi(\bm \beta) \qquad \mbox{with} \qquad \left\{
\begin{array}{l}
\disp f(x) =  \frac{1.442853 - (-x)^{0.4054651}}{0.1183375}, \\
\\
\disp f_\xi(x_1,x_2,x_3) = \frac{2 (x_1 - x_2)}{x_1 - \frac{9}{4} x_3},
\end{array}
\right. \\
\label{eq:tilde_h_sigma}
\sigma &= \tilde h_\sigma(\bm \beta) = \frac{2^{3 - \xi} (\beta_{1} - \beta_{2})}{\Gamma(2 - \xi)}, \\
\label{eq:tilde_h_mu}
\mu &= \tilde h_\mu(\bm \beta) = 4 \beta_{1} + \frac{\sigma}{\xi} \{1 -  2^\xi \Gamma(2 - \xi) \}.
\end{align}
As in Section~\ref{sec:practical}, approximations of the partial derivatives of $h_\xi$, $h_\sigma$ and $h_\mu$ which are necessary for computing $\hat \sigma_{h,n}^2$, are obtained by differentiating $\tilde h_\xi$, $\tilde h_\sigma$ and $\tilde h_\mu$.

\section{Monte-Carlo experiments}
\label{sec:sims}

Monte-Carlo experiments were carried out in order to investigate the finite-sample behavior of the tests studied in the previous sections. Nine statistics were considered: $S_{g,n}$ in~\eqref{eq:S_g_n} and $T_{g,n}$ in~\eqref{eq:T_g_n} with $g$ being respectively $\tilde g_\xi$ in~\eqref{eq:tilde_g_xi}, $\tilde g_\sigma$ in~\eqref{eq:tilde_g_sigma} and $\tilde g_\mu$ in~\eqref{eq:tilde_g_mu}, and the statistic $S_{h,n}$ in~\eqref{eq:S_h_n} with $h$ being respectively $\tilde h_\xi$ in~\eqref{eq:tilde_h_xi}, $\tilde h_\sigma$ in~\eqref{eq:tilde_h_sigma} and $\tilde h_\mu$ in~\eqref{eq:tilde_h_mu}. Recall that, under the conditions of Proposition~\ref{prop:weak_S_g_n}, the statistics $S_{g,n}$ and $T_{g,n}$ are asymptotically equivalent under the null, and that the functions $\tilde g_\xi$, $\tilde g_\sigma$, $\tilde g_\mu$ (resp.\ $\tilde h_\xi$, $\tilde h_\sigma$, $\tilde h_\mu$) are related to the PWM (resp.\ GPWM) method for estimating the GEV parameters summarized in Section~\ref{sec:PWM_method} (resp.\ Section~\ref{sec:GPWM}). When computing the statistics, the following asymptotically negligible change was made: given a small integer $r \geq 1$, maxima in~\eqref{eq:S_g_n},~\eqref{eq:T_g_n} and~\eqref{eq:S_h_n} were taken over the set $\{r,\dots,n-r\}$ instead of the set $\{1,\dots,n-1\}$ or $\{3,\dots,n-3\}$. The role of the integer $r$ is to exclude subsamples of cardinality strictly smaller than $r$. In our simulations, we used the value $r=10$ and, in all the considered settings, encountered hardly any sample for which the indicators in~\eqref{eq:S_g_n} and~\eqref{eq:S_h_n} turned to be zero for some value of $k$.

We started our experiments by investigating the empirical levels of the nine tests for samples of size $n \in \{50,100,200,400\}$ generated from the GEV($\mu$,$\sigma$,$\xi$). To estimate the power of the tests, 1000 samples were generated under each data generating scenario and all the tests were carried out at the 5\% significance level. Our main finding was that all tests had a tendency of being too liberal when $|\mu/\sigma|$ was large. This prompted us to run the tests based on $S_{g,n}$ in~\eqref{eq:S_g_n} (resp.\ $T_{g,n}$ in~\eqref{eq:T_g_n}, $S_{h,n}$ in~\eqref{eq:S_h_n}) on the data translated by $- \tilde g_\mu(\hat{\bm \beta}_{1:n})$ (resp.\ $- \tilde g_\mu(\hat{\bm b}_{1:n})$, $- \tilde h_\mu(\hat{\bm \beta}_{1:n})$). Since $\tilde g_\mu(\hat{\bm \beta}_{1:n}) \as \tilde g_\mu(\bm \beta)$ and $\tilde g_\mu(\hat{\bm b}_{1:n}) \as \tilde g_\mu(\bm \beta)$ under Condition~\ref{cond:PWM}, the asymptotic results stated in Section~\ref{sec:limiting} remain valid when the tests based on the statistics $S_{g,n}$ and $T_{g,n}$ are applied on the translated data. 

The adopted translation of the data led to a significant improvement of the empirical levels. Additional improvement of the levels of the test based on $T_{\tilde g_\sigma,n}$ (resp.\ $T_{\tilde g_\xi,n}$) was obtained by multiplying $\hat \sigma^2_{\tilde g_\sigma,n}$ (resp.\ $\hat \sigma^2_{\tilde g_\xi,n}$) in~\eqref{eq:hat_sigma_g_n} by $(n+10)/n$ (resp. $(n+20)/n$). A last asymptotically negligible change was carried out for the tests based on $S_{h,n}$: the statistics given in~\eqref{eq:S_h_n} were computed with $\gamma$ in~\eqref{eq:Fkl} equal to zero. 

As the tests based on $S_{g,n}$ in~\eqref{eq:S_g_n} turned out to be consistently worse-behaved than those based on $T_{g,n}$ in~\eqref{eq:T_g_n}, for the sake of brevity, we do not report any results for the former in the forthcoming tables.

Table~\ref{resH0_mu_0} gives the empirical levels of the six remaining tests for samples generated from the GEV($\mu$,1,$\xi$) for $\mu = 0$ and $\xi \in \{-1,-0.9,\dots,1,1.1\}$. Analogue tables for $\mu \in \{-100,-10,10,100\}$ can be found in Section~1.1 
of the supplementary material (see Tables~1--4 
therein).  For the sample sizes under consideration, the tests based on $T_{g,n}$ appear to hold their level well for $-1 \leq \xi \leq 0.4$, that is, as long as the variance of the GEV is finite. The tests based on $S_{h,n}$ are overall too conservative for $-1 \leq \xi \leq 0.3$. Their empirical levels are closer to the 5\% nominal level for $0.4 \leq \xi \leq 0.8$ and they start being too liberal for $\xi \geq 0.9$. 

A similar experiment was carried out for samples of independent block maxima obtained from blocks of size $b \in \{1,5,10,50\}$ from the generalized Pareto distribution (GPD) with location parameter equal to 0, scale parameter equal to 1 and shape parameter equal to $\xi$ (abbreviated as GPD(0,1,$\xi$) as we continue). The experiment was designed having in mind the fact that the GPD with shape parameter $\xi$ is in the maximum domain of attraction of the GEV with shape parameter $\xi$. The results for $b=1$ are reported in Table~\ref{resH0_bm_1}. The results for $b \in \{5,10,50\}$ are reported in Section~1.2 
of the supplementary material (see Tables~5--7 
therein). Another variant of this experiment was carried out by using the absolute value of the standard Student $t$ distribution with $1/\xi$ degrees of freedom instead of the GPD(0,1,$\xi$). Indeed, the absolute value of the standard Student $t$ distribution with $1/\xi$ degrees of freedom, $\xi > 0$, is in the maximum domain of attraction of the GEV with shape parameter $\xi$. The results are reported in Section~1.3 
of the supplementary material (see Tables~8--11 
therein). The conclusions in terms of $\xi$ are overall the same as for the first experiment: the tests based on $T_{g,n}$ behave mostly adequately when $\xi \leq 0.4$ and become too liberal for $\xi \geq 0.5$; the tests based on $S_{h,n}$ are overall too conservative for $\xi \leq 0.3$ and are too liberal for $\xi \geq 0.9$. 

\begin{table}[t!]
\centering
\caption{Percentage of rejection of $H_0$ computed from 1000 samples of size $n \in \{50, 100, 200, 400\}$ generated from the GEV(0,1,$\xi$).} 
\label{resH0_mu_0}
\begingroup\small
\begin{tabular}{rrrrrrrrrrrrrrrr}
  \hline
  \multicolumn{2}{c}{} & \multicolumn{3}{c}{PWM / $T_{g,n}$} & \multicolumn{3}{c}{GPWM / $S_{h,n}$} & \multicolumn{2}{c}{} & \multicolumn{3}{c}{PWM / $T_{g,n}$} & \multicolumn{3}{c}{GPWM / $S_{h,n}$}\\ \cmidrule(lr){3-5} \cmidrule(lr){6-8} \cmidrule(lr){11-13} \cmidrule(lr){14-16} $\xi$ & $n$ & $\tilde g_\mu$ & $\tilde g_\sigma$ & $\tilde g_\xi$ & $\tilde h_\mu$ & $\tilde h_\sigma$ & $\tilde h_\xi$ & $\xi$ & $n$ & $\tilde g_\mu$ & $\tilde g_\sigma$ & $\tilde g_\xi$ & $\tilde h_\mu$ & $\tilde h_\sigma$ & $\tilde h_\xi$ \\ \hline
-1.0 & 50 & 3.0 & 0.7 & 4.1 & 7.2 & 0.8 & 0.2 & 0.1 & 50 & 3.8 & 3.2 & 4.0 & 1.2 & 0.6 & 1.4 \\ 
   & 100 & 2.9 & 1.6 & 3.7 & 7.4 & 1.7 & 0.2 &  & 100 & 4.9 & 3.4 & 2.7 & 2.7 & 1.4 & 1.3 \\ 
   & 200 & 4.0 & 2.9 & 4.1 & 5.4 & 3.2 & 0.6 &  & 200 & 4.7 & 4.2 & 2.5 & 3.8 & 2.4 & 2.9 \\ 
   & 400 & 3.8 & 3.3 & 4.4 & 3.9 & 2.9 & 1.1 &  & 400 & 4.4 & 4.5 & 3.4 & 4.1 & 3.5 & 3.9 \\ 
  -0.9 & 50 & 2.6 & 0.7 & 3.8 & 4.8 & 0.2 & 0.4 & 0.2 & 50 & 6.2 & 5.4 & 4.4 & 2.3 & 1.5 & 1.5 \\ 
   & 100 & 4.3 & 2.0 & 4.2 & 6.3 & 1.5 & 0.2 &  & 100 & 5.2 & 4.9 & 3.6 & 3.6 & 2.2 & 1.6 \\ 
   & 200 & 4.8 & 2.8 & 5.1 & 4.8 & 2.4 & 0.7 &  & 200 & 5.0 & 4.4 & 3.2 & 4.3 & 2.6 & 4.2 \\ 
   & 400 & 4.0 & 3.3 & 3.6 & 3.4 & 3.5 & 1.1 &  & 400 & 4.6 & 4.2 & 4.2 & 4.2 & 3.9 & 4.4 \\ 
  -0.8 & 50 & 2.2 & 0.8 & 4.0 & 4.5 & 0.2 & 0.3 & 0.3 & 50 & 5.7 & 3.7 & 5.0 & 2.2 & 0.8 & 1.4 \\ 
   & 100 & 2.2 & 2.4 & 2.8 & 4.3 & 1.6 & 0.1 &  & 100 & 4.3 & 4.1 & 2.9 & 2.6 & 1.4 & 2.0 \\ 
   & 200 & 4.1 & 3.8 & 3.8 & 4.7 & 3.6 & 0.5 &  & 200 & 6.1 & 6.0 & 2.5 & 4.5 & 3.0 & 3.2 \\ 
   & 400 & 5.6 & 4.1 & 4.2 & 5.9 & 4.3 & 2.0 &  & 400 & 5.0 & 4.8 & 2.6 & 3.7 & 3.7 & 4.0 \\ 
  -0.7 & 50 & 2.4 & 1.4 & 3.9 & 3.6 & 0.1 & 0.0 & 0.4 & 50 & 6.6 & 4.4 & 8.1 & 3.1 & 1.5 & 2.1 \\ 
   & 100 & 3.1 & 2.5 & 4.3 & 4.2 & 1.1 & 0.1 &  & 100 & 6.8 & 5.6 & 4.5 & 4.5 & 2.0 & 2.9 \\ 
   & 200 & 3.9 & 3.7 & 4.1 & 5.0 & 2.7 & 0.6 &  & 200 & 6.1 & 5.0 & 3.3 & 4.0 & 3.8 & 3.0 \\ 
   & 400 & 3.7 & 3.7 & 3.9 & 4.3 & 3.5 & 1.3 &  & 400 & 5.3 & 4.5 & 2.1 & 3.6 & 4.8 & 4.3 \\ 
  -0.6 & 50 & 2.2 & 1.9 & 3.6 & 2.2 & 0.1 & 0.3 & 0.5 & 50 & 8.0 & 4.2 & 9.5 & 2.7 & 1.3 & 3.5 \\ 
   & 100 & 4.1 & 2.9 & 3.8 & 3.7 & 0.7 & 0.3 &  & 100 & 9.2 & 7.0 & 6.8 & 4.8 & 3.2 & 2.9 \\ 
   & 200 & 5.1 & 4.1 & 3.5 & 4.3 & 2.5 & 0.7 &  & 200 & 5.8 & 4.5 & 3.4 & 4.8 & 3.3 & 3.5 \\ 
   & 400 & 4.3 & 4.6 & 4.0 & 4.4 & 4.1 & 2.4 &  & 400 & 5.4 & 5.6 & 2.9 & 4.5 & 4.1 & 4.2 \\ 
  -0.5 & 50 & 3.3 & 3.2 & 4.5 & 2.7 & 0.3 & 0.7 & 0.6 & 50 & 6.7 & 4.8 & 16.6 & 4.2 & 2.5 & 3.4 \\ 
   & 100 & 3.4 & 2.3 & 4.7 & 3.7 & 0.5 & 0.7 &  & 100 & 8.5 & 6.8 & 10.9 & 5.6 & 4.4 & 3.0 \\ 
   & 200 & 3.4 & 3.8 & 4.9 & 3.8 & 2.2 & 1.3 &  & 200 & 6.4 & 5.9 & 7.6 & 4.5 & 3.9 & 3.7 \\ 
   & 400 & 4.5 & 4.4 & 4.1 & 4.5 & 3.3 & 2.4 &  & 400 & 5.3 & 7.6 & 2.7 & 3.8 & 4.8 & 3.4 \\ 
  -0.4 & 50 & 2.5 & 3.2 & 4.9 & 2.0 & 0.4 & 0.3 & 0.7 & 50 & 7.8 & 4.0 & 18.5 & 3.9 & 2.1 & 3.2 \\ 
   & 100 & 3.6 & 3.5 & 3.9 & 2.8 & 1.2 & 1.0 &  & 100 & 10.4 & 7.2 & 17.1 & 5.8 & 5.1 & 2.4 \\ 
   & 200 & 3.0 & 5.0 & 3.9 & 3.2 & 2.7 & 1.5 &  & 200 & 7.1 & 6.6 & 8.5 & 5.3 & 4.7 & 3.2 \\ 
   & 400 & 5.3 & 4.9 & 3.7 & 5.2 & 3.5 & 1.9 &  & 400 & 7.4 & 6.9 & 6.5 & 6.3 & 5.4 & 3.5 \\ 
  -0.3 & 50 & 2.8 & 4.2 & 4.4 & 1.3 & 0.7 & 0.6 & 0.8 & 50 & 9.6 & 3.1 & 27.4 & 4.6 & 4.3 & 5.9 \\ 
   & 100 & 4.6 & 3.0 & 4.3 & 3.9 & 0.9 & 0.7 &  & 100 & 10.1 & 7.3 & 24.3 & 7.2 & 6.0 & 3.2 \\ 
   & 200 & 4.6 & 4.2 & 4.6 & 4.4 & 1.9 & 2.1 &  & 200 & 9.4 & 7.7 & 18.8 & 6.8 & 6.3 & 2.7 \\ 
   & 400 & 4.2 & 4.1 & 4.7 & 3.6 & 3.0 & 3.1 &  & 400 & 7.5 & 7.4 & 10.6 & 5.4 & 5.1 & 3.0 \\ 
  -0.2 & 50 & 2.6 & 2.9 & 3.0 & 1.5 & 0.4 & 0.4 & 0.9 & 50 & 11.9 & 5.5 & 31.1 & 7.0 & 5.3 & 6.1 \\ 
   & 100 & 5.1 & 4.3 & 3.9 & 3.6 & 0.9 & 1.0 &  & 100 & 10.8 & 7.4 & 33.6 & 7.8 & 6.0 & 2.2 \\ 
   & 200 & 4.7 & 4.6 & 4.3 & 4.0 & 2.7 & 2.0 &  & 200 & 10.4 & 9.3 & 25.5 & 7.4 & 6.0 & 3.0 \\ 
   & 400 & 5.1 & 4.1 & 4.6 & 3.9 & 3.0 & 3.2 &  & 400 & 7.6 & 7.4 & 18.9 & 6.0 & 6.7 & 2.7 \\ 
  -0.1 & 50 & 4.0 & 3.8 & 3.9 & 1.8 & 0.7 & 0.8 & 1.0 & 50 & 12.0 & 4.6 & 37.8 & 7.6 & 6.4 & 9.5 \\ 
   & 100 & 4.2 & 4.2 & 4.7 & 2.6 & 1.6 & 1.5 &  & 100 & 11.6 & 7.4 & 39.1 & 7.4 & 6.7 & 4.2 \\ 
   & 200 & 4.7 & 4.4 & 4.5 & 4.1 & 2.4 & 2.1 &  & 200 & 12.3 & 8.5 & 32.4 & 9.6 & 10.1 & 3.0 \\ 
   & 400 & 4.0 & 3.7 & 4.2 & 3.4 & 2.4 & 3.8 &  & 400 & 7.8 & 5.2 & 27.5 & 5.8 & 6.2 & 2.8 \\ 
  0.0 & 50 & 4.7 & 3.0 & 4.0 & 1.8 & 0.3 & 1.3 & 1.1 & 50 & 13.9 & 4.8 & 38.8 & 10.8 & 7.6 & 12.2 \\ 
   & 100 & 3.6 & 3.3 & 2.8 & 2.6 & 1.3 & 1.5 &  & 100 & 15.3 & 8.3 & 44.9 & 12.0 & 9.6 & 5.3 \\ 
   & 200 & 4.5 & 3.6 & 3.8 & 3.6 & 2.2 & 1.7 &  & 200 & 9.8 & 7.4 & 44.4 & 7.9 & 9.5 & 2.3 \\ 
   & 400 & 4.6 & 3.4 & 3.7 & 4.0 & 2.8 & 4.0 &  & 400 & 7.5 & 6.4 & 38.0 & 7.0 & 8.3 & 2.4 \\ 
   \hline
\end{tabular}
\endgroup
\end{table}

\begin{table}[t!]
\centering
\caption{Percentage of rejection of $H_0$ computed from 1000 samples of $n \in \{50, 100, 200, 400\}$ independent block maxima obtained from blocks of size 1 from the GPD(0,1,$\xi$).} 
\label{resH0_bm_1}
\begingroup\small
\begin{tabular}{rrrrrrrrrrrrrrrr}
  \hline
  \multicolumn{2}{c}{} & \multicolumn{3}{c}{PWM / $T_{g,n}$} & \multicolumn{3}{c}{GPWM / $S_{h,n}$} & \multicolumn{2}{c}{} & \multicolumn{3}{c}{PWM / $T_{g,n}$} & \multicolumn{3}{c}{GPWM / $S_{h,n}$}\\ \cmidrule(lr){3-5} \cmidrule(lr){6-8} \cmidrule(lr){11-13} \cmidrule(lr){14-16} $\xi$ & $n$ & $\tilde g_\mu$ & $\tilde g_\sigma$ & $\tilde g_\xi$ & $\tilde h_\mu$ & $\tilde h_\sigma$ & $\tilde h_\xi$ & $\xi$ & $n$ & $\tilde g_\mu$ & $\tilde g_\sigma$ & $\tilde g_\xi$ & $\tilde h_\mu$ & $\tilde h_\sigma$ & $\tilde h_\xi$ \\ \hline
-1.0 & 50 & 3.0 & 3.9 & 1.7 & 1.6 & 3.0 & 1.8 & 0.1 & 50 & 6.1 & 2.9 & 5.5 & 3.5 & 1.5 & 5.0 \\ 
   & 100 & 3.7 & 4.3 & 2.5 & 2.2 & 2.8 & 1.6 &  & 100 & 5.7 & 3.3 & 4.3 & 3.9 & 2.2 & 4.2 \\ 
   & 200 & 3.8 & 4.3 & 3.1 & 2.7 & 3.6 & 2.9 &  & 200 & 4.9 & 3.3 & 2.4 & 3.8 & 2.9 & 3.5 \\ 
   & 400 & 4.2 & 4.7 & 3.7 & 3.9 & 4.3 & 3.1 &  & 400 & 5.5 & 3.5 & 4.2 & 4.8 & 3.7 & 4.9 \\ 
  -0.9 & 50 & 3.4 & 3.5 & 2.8 & 1.5 & 3.0 & 1.8 & 0.2 & 50 & 6.6 & 3.4 & 5.0 & 2.6 & 1.5 & 3.9 \\ 
   & 100 & 4.9 & 3.5 & 3.8 & 2.8 & 1.8 & 2.5 &  & 100 & 5.4 & 3.6 & 3.7 & 3.6 & 2.4 & 4.3 \\ 
   & 200 & 4.4 & 4.4 & 3.5 & 3.6 & 3.7 & 3.2 &  & 200 & 5.7 & 4.0 & 3.1 & 5.0 & 4.5 & 4.8 \\ 
   & 400 & 3.5 & 4.3 & 3.5 & 3.3 & 3.6 & 4.2 &  & 400 & 5.0 & 5.3 & 3.1 & 3.8 & 4.6 & 4.4 \\ 
  -0.8 & 50 & 3.8 & 1.6 & 2.3 & 1.1 & 1.3 & 1.7 & 0.3 & 50 & 6.3 & 3.7 & 6.3 & 3.0 & 1.8 & 5.2 \\ 
   & 100 & 3.5 & 2.7 & 3.0 & 2.7 & 1.9 & 2.0 &  & 100 & 6.5 & 4.8 & 5.5 & 4.4 & 2.9 & 4.6 \\ 
   & 200 & 4.4 & 3.7 & 3.5 & 3.5 & 2.9 & 3.4 &  & 200 & 5.7 & 5.0 & 2.3 & 4.0 & 3.4 & 4.7 \\ 
   & 400 & 4.8 & 4.2 & 4.3 & 4.4 & 3.5 & 4.5 &  & 400 & 5.2 & 4.1 & 2.8 & 4.3 & 3.8 & 4.5 \\ 
  -0.7 & 50 & 4.5 & 1.9 & 2.0 & 1.9 & 1.0 & 1.8 & 0.4 & 50 & 9.8 & 5.1 & 9.0 & 5.7 & 2.9 & 4.3 \\ 
   & 100 & 5.2 & 2.5 & 3.2 & 3.5 & 2.0 & 1.5 &  & 100 & 7.8 & 6.1 & 5.0 & 5.2 & 3.0 & 4.1 \\ 
   & 200 & 4.3 & 3.3 & 3.8 & 2.9 & 2.1 & 3.8 &  & 200 & 7.3 & 6.0 & 3.3 & 6.1 & 5.0 & 3.6 \\ 
   & 400 & 3.4 & 4.4 & 5.3 & 3.4 & 3.9 & 4.6 &  & 400 & 5.9 & 4.3 & 2.8 & 5.4 & 3.9 & 4.4 \\ 
  -0.6 & 50 & 4.1 & 2.2 & 2.5 & 2.0 & 1.8 & 2.4 & 0.5 & 50 & 7.8 & 3.9 & 11.0 & 4.3 & 2.0 & 3.8 \\ 
   & 100 & 4.5 & 3.1 & 3.3 & 3.5 & 2.2 & 3.8 &  & 100 & 7.0 & 6.8 & 9.2 & 5.1 & 4.9 & 4.5 \\ 
   & 200 & 4.0 & 3.2 & 3.3 & 2.8 & 2.9 & 2.9 &  & 200 & 6.8 & 5.0 & 6.2 & 4.8 & 3.1 & 4.3 \\ 
   & 400 & 4.1 & 4.5 & 3.4 & 4.1 & 3.8 & 4.0 &  & 400 & 6.5 & 5.9 & 3.5 & 4.3 & 4.6 & 3.3 \\ 
  -0.5 & 50 & 4.6 & 1.9 & 2.3 & 2.1 & 1.2 & 1.8 & 0.6 & 50 & 8.4 & 4.5 & 17.0 & 5.1 & 2.8 & 5.5 \\ 
   & 100 & 4.1 & 2.7 & 2.4 & 2.3 & 1.9 & 2.5 &  & 100 & 10.1 & 7.1 & 11.1 & 6.7 & 4.9 & 4.9 \\ 
   & 200 & 6.3 & 3.6 & 5.4 & 4.5 & 3.6 & 3.7 &  & 200 & 9.3 & 7.6 & 7.0 & 6.2 & 5.3 & 3.6 \\ 
   & 400 & 5.7 & 2.9 & 4.9 & 5.0 & 3.0 & 5.6 &  & 400 & 5.9 & 6.4 & 4.8 & 5.1 & 5.3 & 3.8 \\ 
  -0.4 & 50 & 5.0 & 2.1 & 3.7 & 1.9 & 1.0 & 3.1 & 0.7 & 50 & 9.9 & 5.7 & 22.1 & 6.2 & 3.9 & 6.7 \\ 
   & 100 & 4.5 & 1.7 & 3.5 & 2.7 & 1.5 & 3.1 &  & 100 & 9.9 & 7.0 & 19.8 & 7.2 & 5.2 & 3.5 \\ 
   & 200 & 4.6 & 3.3 & 3.3 & 3.7 & 2.9 & 3.4 &  & 200 & 8.0 & 7.5 & 11.2 & 6.6 & 5.4 & 3.2 \\ 
   & 400 & 5.2 & 3.8 & 3.8 & 4.6 & 3.8 & 3.3 &  & 400 & 7.1 & 5.5 & 5.2 & 6.5 & 4.4 & 4.4 \\ 
  -0.3 & 50 & 3.3 & 1.9 & 3.6 & 1.3 & 0.6 & 1.8 & 0.8 & 50 & 11.4 & 5.9 & 27.4 & 7.1 & 5.5 & 6.9 \\ 
   & 100 & 5.1 & 3.4 & 3.2 & 2.9 & 2.2 & 3.2 &  & 100 & 13.2 & 7.1 & 25.2 & 8.8 & 6.6 & 3.8 \\ 
   & 200 & 5.4 & 2.7 & 3.0 & 3.8 & 2.8 & 4.0 &  & 200 & 10.4 & 7.2 & 16.7 & 7.7 & 7.5 & 2.5 \\ 
   & 400 & 6.0 & 3.6 & 4.1 & 5.8 & 3.4 & 4.1 &  & 400 & 6.3 & 6.2 & 12.0 & 5.1 & 4.8 & 3.2 \\ 
  -0.2 & 50 & 4.9 & 1.4 & 3.5 & 1.8 & 0.8 & 2.1 & 0.9 & 50 & 12.4 & 5.9 & 31.0 & 9.9 & 7.7 & 9.4 \\ 
   & 100 & 4.5 & 2.6 & 3.3 & 3.0 & 1.3 & 2.8 &  & 100 & 12.4 & 8.6 & 32.0 & 7.9 & 8.9 & 5.2 \\ 
   & 200 & 5.2 & 3.1 & 5.0 & 4.0 & 2.7 & 4.1 &  & 200 & 10.2 & 6.8 & 25.0 & 8.1 & 7.2 & 1.7 \\ 
   & 400 & 4.2 & 3.9 & 3.2 & 4.7 & 4.1 & 4.0 &  & 400 & 8.1 & 7.0 & 18.5 & 6.3 & 5.6 & 2.4 \\ 
  -0.1 & 50 & 6.8 & 2.8 & 3.8 & 4.0 & 0.9 & 3.5 & 1.0 & 50 & 12.6 & 4.9 & 37.0 & 8.5 & 7.0 & 9.1 \\ 
   & 100 & 5.1 & 2.2 & 3.5 & 3.3 & 1.8 & 4.0 &  & 100 & 13.6 & 6.8 & 40.4 & 11.1 & 8.8 & 4.8 \\ 
   & 200 & 4.6 & 3.7 & 3.0 & 4.2 & 3.0 & 3.8 &  & 200 & 10.9 & 7.4 & 35.1 & 9.7 & 9.5 & 2.1 \\ 
   & 400 & 5.5 & 4.6 & 5.5 & 5.3 & 3.9 & 4.8 &  & 400 & 8.0 & 7.4 & 25.6 & 6.7 & 7.3 & 3.1 \\ 
  0.0 & 50 & 6.1 & 3.5 & 2.9 & 2.9 & 1.8 & 3.7 & 1.1 & 50 & 16.3 & 5.6 & 40.0 & 11.7 & 9.2 & 14.3 \\ 
   & 100 & 5.0 & 3.1 & 4.4 & 3.6 & 1.5 & 4.2 &  & 100 & 15.6 & 8.3 & 47.3 & 10.8 & 11.3 & 5.7 \\ 
   & 200 & 6.2 & 3.7 & 4.3 & 5.1 & 4.0 & 4.1 &  & 200 & 11.2 & 8.3 & 41.5 & 11.0 & 9.6 & 3.7 \\ 
   & 400 & 5.7 & 4.7 & 4.5 & 5.1 & 4.0 & 5.1 &  & 400 & 8.7 & 7.7 & 37.1 & 8.0 & 7.4 & 2.5 \\ 
   \hline
\end{tabular}
\endgroup
\end{table}

Additional experiments under the null (whose results are reported in Table~12 
of the supplementary material) were carried out to assess the empirical levels of the tests for samples generated from various normal and exponential distributions. No test turned out to be too liberal. 

The next experiments focused on the power of the tests. To ease reading of the forthcoming tables, among the six tests proposed in the paper, the rejection percentages of those that are expected to be sensitive to the alternative under consideration are colored in light gray.

We started by assessing the power of the tests in the case of a change in the shape parameter of the GEV. First, for $t \in \{0.25,0.5,0.75\}$, samples of size $n \in \{100, 200\}$ were generated such that the $\ip{nt}$ first observations were from a GEV($\mu$,1,-0.4) and the $n-\ip{nt}$ last observations were from a GEV($\mu$,1,$\xi$), for $\xi \in \{-0.2,0,0.2,0.4\}$ and $\mu \in \{-100,0,100\}$. The results for $\mu=0$ are reported in Table~\ref{resH1xi1_-0.4_mu_0} while those for $\mu \in \{-100,100\}$ are given in Tables~13 
and~14 
of the supplementary material. The column $F_n$ gives the rejection percentages obtained using the nonparametric rank-based test for change-point detection based on empirical c.d.f.s studied in \cite{HolKojQue13} and implemented in the function {\tt cpTestFn} of the \textsf{R} package {\tt npcp}. The columns $\bar x_n$ and $s_n^2$ report the empirical powers of CUSUM tests designed to be particularly sensitive to changes in the expectation \citep[see][and the references therein]{Phi87} and the variance \citep[see, e.g.,][]{BucKoj16b}, respectively. The former (resp.\ latter) is implemented in the function {\tt cpTestMean} (resp.\ {\tt cpTestU}) of the \textsf{R} package {\tt npcp}. Tables~\ref{resH1xi1_0_mu_0} and~\ref{resH1xi1_0.2_mu_0} report similar results but for higher $\xi$ values. Overall, as could have been expected, the tests based on $T_{\tilde g_\mu,n}$, $T_{\tilde g_\sigma}$, $S_{\tilde h_\mu,n}$ and $S_{\tilde h_\sigma,n}$ have hardly any power against such alternatives. Roughly speaking, the test based on $T_{\tilde g_\xi,n}$ is more (resp.\ less) powerful than the one based on $S_{\tilde h_\xi,n}$ when the largest shape parameter is smaller or equal than $0.2$ (resp.\ greater or equal than $0.4$). Both tests are overall more powerful than the general-purpose nonparametric test considered in \cite{HolKojQue13} and the two CUSUM tests designed to be particularly sensitive to changes in the expectation and the variance, respectively.

\begin{table}[t!]
\centering
\caption{Percentage of rejection of $H_0$ computed from 1000 samples of size $n \in \{100, 200\}$ such that the $\ip{nt}$ first observations are from a GEV(0,1,-0.4) and the $n - \ip{nt}$ last observations are from a GEV(0,1,$\xi$).} 
\label{resH1xi1_-0.4_mu_0}
\begingroup\small
\begin{tabular}{rrrrrrrrarra}
  \hline
  \multicolumn{6}{c}{} & \multicolumn{3}{c}{PWM / $T_{g,n}$} & \multicolumn{3}{c}{GPWM / $S_{h,n}$} \\ \cmidrule(lr){7-9} \cmidrule(lr){10-12} $\xi$ & $n$ & $t$ & $F_n$ & $\bar x_n$ & $s_n^2$ & $\tilde g_\mu$ & $\tilde g_\sigma$ & $\tilde g_\xi$  & $\tilde h_\mu$ & $\tilde h_\sigma$ & $\tilde h_\xi$  \\ \hline
-0.2 & 100 & 0.25 & 5.7 & 5.1 & 4.3 & 3.2 & 3.8 & 10.2 & 2.3 & 1.7 & 1.2 \\ 
   &  & 0.75 & 6.1 & 4.9 & 7.0 & 4.5 & 3.4 & 5.1 & 3.3 & 0.9 & 1.4 \\ 
   &  & 0.50 & 5.5 & 6.5 & 8.1 & 4.2 & 3.4 & 12.5 & 2.9 & 1.7 & 2.0 \\ 
   & 200 & 0.25 & 5.2 & 7.0 & 6.0 & 3.7 & 4.0 & 17.7 & 3.7 & 2.2 & 5.1 \\ 
   &  & 0.75 & 7.3 & 8.3 & 8.5 & 5.5 & 3.9 & 15.3 & 4.5 & 1.9 & 3.5 \\ 
   &  & 0.50 & 8.5 & 11.7 & 10.7 & 4.9 & 3.5 & 29.4 & 4.7 & 1.6 & 6.9 \\ 
  0.0 & 100 & 0.25 & 5.8 & 7.6 & 4.5 & 4.8 & 4.2 & 23.1 & 2.6 & 1.2 & 3.2 \\ 
   &  & 0.75 & 7.6 & 12.6 & 13.0 & 5.3 & 4.4 & 11.6 & 4.2 & 1.1 & 4.5 \\ 
   &  & 0.50 & 7.8 & 16.8 & 10.2 & 4.4 & 2.6 & 32.8 & 3.0 & 0.8 & 9.1 \\ 
   & 200 & 0.25 & 8.3 & 16.8 & 7.6 & 4.5 & 2.6 & 54.4 & 3.8 & 1.1 & 22.7 \\ 
   &  & 0.75 & 8.0 & 22.2 & 29.8 & 5.7 & 5.2 & 40.2 & 3.7 & 2.1 & 15.4 \\ 
   &  & 0.50 & 10.3 & 32.1 & 30.7 & 4.6 & 3.1 & 76.1 & 3.2 & 1.5 & 37.9 \\ 
  0.2 & 100 & 0.25 & 8.0 & 9.7 & 2.4 & 6.1 & 4.7 & 42.7 & 3.1 & 1.4 & 14.3 \\ 
   &  & 0.75 & 8.3 & 23.6 & 11.7 & 5.8 & 5.5 & 16.0 & 4.2 & 1.5 & 9.2 \\ 
   &  & 0.50 & 12.0 & 27.7 & 10.3 & 6.8 & 6.7 & 52.8 & 3.8 & 3.0 & 26.8 \\ 
   & 200 & 0.25 & 12.0 & 24.9 & 4.9 & 7.1 & 4.7 & 78.2 & 5.0 & 2.8 & 58.2 \\ 
   &  & 0.75 & 11.5 & 51.3 & 39.4 & 6.8 & 4.4 & 59.7 & 4.2 & 2.2 & 44.6 \\ 
   &  & 0.50 & 22.0 & 65.9 & 35.0 & 6.6 & 6.4 & 92.4 & 4.9 & 3.2 & 81.8 \\ 
  0.4 & 100 & 0.25 & 9.6 & 11.1 & 0.5 & 9.6 & 7.3 & 62.6 & 4.3 & 2.0 & 27.5 \\ 
   &  & 0.75 & 9.3 & 34.4 & 11.7 & 6.4 & 9.5 & 21.7 & 3.5 & 2.5 & 20.9 \\ 
   &  & 0.50 & 15.7 & 40.8 & 7.2 & 8.1 & 10.1 & 69.5 & 4.2 & 2.5 & 50.3 \\ 
   & 200 & 0.25 & 15.5 & 25.1 & 2.3 & 11.1 & 7.0 & 89.2 & 4.6 & 2.9 & 85.1 \\ 
   &  & 0.75 & 13.7 & 71.4 & 36.3 & 7.9 & 7.6 & 63.9 & 3.3 & 4.2 & 79.2 \\ 
   &  & 0.50 & 29.6 & 80.6 & 26.9 & 12.3 & 8.7 & 95.7 & 6.3 & 3.4 & 97.4 \\ 
   \hline
\end{tabular}
\endgroup
\end{table}

\begin{table}[t!]
\centering
\caption{Percentage of rejection of $H_0$ computed from 1000 samples of size $n \in \{100, 200\}$ such that the $\ip{nt}$ first observations are from a GEV(0,1,0) and the $n - \ip{nt}$ last observations are from a GEV(0,1,$\xi$).} 
\label{resH1xi1_0_mu_0}
\begingroup\small
\begin{tabular}{rrrrrrrrarra}
  \hline
  \multicolumn{6}{c}{} & \multicolumn{3}{c}{PWM / $T_{g,n}$} & \multicolumn{3}{c}{GPWM / $S_{h,n}$} \\ \cmidrule(lr){7-9} \cmidrule(lr){10-12} $\xi$ & $n$ & $t$ & $F_n$ & $\bar x_n$ & $s_n^2$ & $\tilde g_\mu$ & $\tilde g_\sigma$ & $\tilde g_\xi$  & $\tilde h_\mu$ & $\tilde h_\sigma$ & $\tilde h_\xi$  \\ \hline
0.2 & 100 & 0.25 & 5.7 & 4.4 & 1.2 & 5.8 & 4.2 & 6.7 & 2.8 & 2.0 & 4.1 \\ 
   &  & 0.75 & 6.7 & 7.4 & 3.2 & 4.6 & 4.6 & 3.9 & 3.5 & 2.2 & 2.9 \\ 
   &  & 0.50 & 4.1 & 5.5 & 2.2 & 4.6 & 4.8 & 6.7 & 2.9 & 1.3 & 4.3 \\ 
   & 200 & 0.25 & 8.2 & 6.7 & 2.2 & 5.9 & 4.4 & 10.0 & 4.8 & 2.4 & 8.8 \\ 
   &  & 0.75 & 5.7 & 9.9 & 7.8 & 5.7 & 4.8 & 6.0 & 4.4 & 3.3 & 7.2 \\ 
   &  & 0.50 & 6.1 & 9.9 & 7.3 & 5.2 & 4.8 & 14.7 & 4.6 & 2.8 & 15.9 \\ 
  0.4 & 100 & 0.25 & 5.9 & 4.6 & 0.4 & 6.2 & 4.4 & 15.5 & 3.3 & 2.1 & 7.0 \\ 
   &  & 0.75 & 6.1 & 11.3 & 2.5 & 7.0 & 8.0 & 5.9 & 3.8 & 1.7 & 6.6 \\ 
   &  & 0.50 & 6.4 & 13.1 & 4.4 & 6.7 & 8.0 & 19.2 & 2.9 & 2.2 & 17.4 \\ 
   & 200 & 0.25 & 7.1 & 8.0 & 1.8 & 5.4 & 5.2 & 28.8 & 3.9 & 2.4 & 27.9 \\ 
   &  & 0.75 & 7.7 & 25.9 & 15.2 & 6.1 & 7.2 & 13.0 & 4.5 & 3.2 & 24.5 \\ 
   &  & 0.50 & 7.6 & 30.6 & 10.3 & 5.1 & 6.4 & 37.9 & 3.7 & 2.8 & 54.2 \\ 
  0.6 & 100 & 0.25 & 6.7 & 5.1 & 0.7 & 8.6 & 7.9 & 39.9 & 4.5 & 3.2 & 17.8 \\ 
   &  & 0.75 & 7.3 & 19.5 & 5.2 & 7.5 & 13.7 & 13.1 & 3.9 & 3.6 & 15.0 \\ 
   &  & 0.50 & 7.7 & 17.4 & 1.6 & 7.4 & 12.5 & 36.8 & 3.8 & 3.4 & 35.1 \\ 
   & 200 & 0.25 & 8.4 & 9.4 & 1.3 & 9.1 & 8.8 & 57.6 & 5.2 & 3.9 & 57.7 \\ 
   &  & 0.75 & 8.3 & 47.4 & 21.9 & 7.1 & 15.0 & 24.3 & 4.0 & 5.6 & 56.6 \\ 
   &  & 0.50 & 15.4 & 45.7 & 11.5 & 10.6 & 14.7 & 68.3 & 5.1 & 5.7 & 84.1 \\ 
   \hline
\end{tabular}
\endgroup
\end{table}

\begin{table}[t!]
\centering
\caption{Percentage of rejection of $H_0$ computed from 1000 samples of size $n \in \{100, 200\}$ such that the $\ip{nt}$ first observations are from a GEV(0,1,0.2) and the $n - \ip{nt}$ last observations are from a GEV(0,1,$\xi$).} 
\label{resH1xi1_0.2_mu_0}
\begingroup\small
\begin{tabular}{rrrrrrrrarra}
  \hline
  \multicolumn{6}{c}{} & \multicolumn{3}{c}{PWM / $T_{g,n}$} & \multicolumn{3}{c}{GPWM / $S_{h,n}$} \\ \cmidrule(lr){7-9} \cmidrule(lr){10-12} $\xi$ & $n$ & $t$ & $F_n$ & $\bar x_n$ & $s_n^2$ & $\tilde g_\mu$ & $\tilde g_\sigma$ & $\tilde g_\xi$  & $\tilde h_\mu$ & $\tilde h_\sigma$ & $\tilde h_\xi$  \\ \hline
0.4 & 100 & 0.25 & 5.7 & 2.8 & 0.5 & 6.8 & 4.7 & 6.1 & 4.0 & 1.2 & 4.5 \\ 
   &  & 0.75 & 6.1 & 4.3 & 1.1 & 6.1 & 5.5 & 3.7 & 4.3 & 2.3 & 3.6 \\ 
   &  & 0.50 & 5.9 & 6.5 & 0.3 & 6.7 & 6.0 & 6.6 & 3.8 & 1.7 & 5.0 \\ 
   & 200 & 0.25 & 5.2 & 4.4 & 0.8 & 5.5 & 4.8 & 7.8 & 4.0 & 3.2 & 8.9 \\ 
   &  & 0.75 & 6.0 & 7.9 & 4.5 & 5.9 & 6.5 & 3.4 & 4.8 & 4.1 & 8.7 \\ 
   &  & 0.50 & 5.4 & 8.9 & 2.9 & 4.3 & 5.0 & 9.7 & 3.2 & 2.3 & 15.9 \\ 
  0.6 & 100 & 0.25 & 6.7 & 2.4 & 0.2 & 8.1 & 6.3 & 22.6 & 4.5 & 2.9 & 10.7 \\ 
   &  & 0.75 & 7.3 & 8.1 & 2.3 & 8.1 & 9.4 & 7.0 & 5.3 & 3.0 & 6.7 \\ 
   &  & 0.50 & 7.6 & 7.8 & 1.3 & 9.7 & 9.4 & 18.4 & 4.6 & 3.5 & 16.3 \\ 
   & 200 & 0.25 & 6.4 & 6.1 & 0.5 & 6.9 & 7.6 & 27.7 & 4.8 & 4.6 & 24.5 \\ 
   &  & 0.75 & 7.0 & 21.0 & 7.4 & 6.7 & 10.4 & 9.4 & 4.5 & 5.7 & 23.5 \\ 
   &  & 0.50 & 7.5 & 19.2 & 4.5 & 7.7 & 10.0 & 29.5 & 4.0 & 3.7 & 48.9 \\ 
  0.8 & 100 & 0.25 & 6.0 & 3.1 & 0.3 & 10.7 & 8.4 & 49.5 & 4.9 & 4.1 & 16.5 \\ 
   &  & 0.75 & 5.8 & 13.9 & 2.5 & 8.8 & 17.0 & 13.6 & 4.4 & 4.1 & 11.1 \\ 
   &  & 0.50 & 8.3 & 13.1 & 1.0 & 12.9 & 18.4 & 43.9 & 5.3 & 5.8 & 28.3 \\ 
   & 200 & 0.25 & 9.2 & 5.6 & 0.8 & 13.4 & 11.3 & 58.6 & 6.5 & 5.6 & 46.4 \\ 
   &  & 0.75 & 7.4 & 38.3 & 11.2 & 9.7 & 23.0 & 23.6 & 4.6 & 7.1 & 49.1 \\ 
   &  & 0.50 & 14.2 & 30.6 & 4.0 & 15.8 & 24.7 & 60.5 & 6.9 & 8.5 & 75.3 \\ 
   \hline
\end{tabular}
\endgroup
\end{table}

Similar experiments were used to assess the power of the tests when the scale (resp.\ location) parameter increases from 0.5 to 1 (resp.\ from 0 to 0.5). The corresponding rejection percentages are reported in Tables~\ref{resH1sigma_mu_0} and~\ref{resH1mu_mu_0}, respectively (additional results are available in Sections~2.2 and 2.3 
of the supplementary material). For changes in the scale parameter, we see that the test based on $T_{\tilde g_\sigma,n}$ is more (resp.\ less) powerful than the test based on $S_{\tilde h_\sigma,n}$ for $\xi \in \{-0.4,0\}$ (resp. $\xi \in \{0.4,0.8\}$), the two tests becoming more powerful than the test ``$s_n^2$'' for $\xi \geq 0$. Again, as expected, the remaining PWM or GPWM tests have little power against such alternatives. For changes in the location parameter, the test based on $T_{\tilde g_\mu,n}$ is more (resp.\ less) powerful than the test based on $S_{\tilde h_\mu,n}$ for $\xi \in \{-0.4,0\}$ (resp. $\xi=0.8$). This time the rank-based test of \cite{HolKojQue13} is more powerful than the two aforementioned tests, although the difference is small for $\xi \in \{0,0.4,0.8\}$.

\begin{table}[t!]
\centering
\caption{Percentage of rejection of $H_0$ computed from 1000 samples of size $n \in \{100, 200\}$ such that the $\ip{nt}$ first observations are from a GEV(0,0.5,$\xi$) and the $n - \ip{nt}$ last observations are from a GEV(0,1,$\xi$).} 
\label{resH1sigma_mu_0}
\begingroup\small
\begin{tabular}{rrrrrrrarrar}
  \hline
  \multicolumn{6}{c}{} & \multicolumn{3}{c}{PWM / $T_{g,n}$} & \multicolumn{3}{c}{GPWM / $S_{h,n}$} \\ \cmidrule(lr){7-9} \cmidrule(lr){10-12} $\xi$ & $n$ & $t$ & $F_n$ & $\bar x_n$ & $s_n$ & $\tilde g_\mu$ & $\tilde g_\sigma$ & $\tilde g_\xi$  & $\tilde h_\mu$ & $\tilde h_\sigma$ & $\tilde h_\xi$  \\ \hline
-0.4 & 100 & 0.25 & 17.0 & 7.3 & 54.3 & 4.5 & 62.2 & 9.8 & 4.8 & 38.1 & 1.8 \\ 
   &  & 0.75 & 20.1 & 13.4 & 94.5 & 13.5 & 82.3 & 2.6 & 13.9 & 58.9 & 0.9 \\ 
   &  & 0.50 & 40.7 & 12.4 & 98.1 & 8.5 & 93.6 & 5.7 & 9.8 & 81.9 & 1.0 \\ 
   & 200 & 0.25 & 50.0 & 8.4 & 97.8 & 5.5 & 97.4 & 11.8 & 4.2 & 93.7 & 3.5 \\ 
   &  & 0.75 & 44.7 & 19.6 & 99.9 & 15.5 & 98.7 & 2.8 & 15.6 & 96.5 & 0.4 \\ 
   &  & 0.50 & 85.6 & 24.4 & 100.0 & 10.5 & 100.0 & 5.1 & 10.9 & 100.0 & 0.9 \\ 
  0.0 & 100 & 0.25 & 15.0 & 9.0 & 18.3 & 6.8 & 57.6 & 6.7 & 5.8 & 39.1 & 2.1 \\ 
   &  & 0.75 & 13.7 & 17.7 & 59.6 & 12.5 & 78.6 & 3.1 & 10.3 & 50.6 & 0.7 \\ 
   &  & 0.50 & 26.0 & 20.7 & 64.7 & 11.0 & 91.7 & 5.5 & 9.9 & 79.1 & 0.5 \\ 
   & 200 & 0.25 & 28.1 & 15.7 & 50.6 & 6.3 & 95.6 & 7.5 & 4.8 & 89.9 & 2.7 \\ 
   &  & 0.75 & 29.5 & 33.1 & 91.0 & 15.0 & 97.9 & 1.7 & 13.9 & 95.8 & 0.7 \\ 
   &  & 0.50 & 63.1 & 41.6 & 94.6 & 10.8 & 99.9 & 4.9 & 10.5 & 99.8 & 0.7 \\ 
  0.4 & 100 & 0.25 & 8.2 & 3.9 & 0.6 & 6.9 & 25.2 & 7.1 & 5.0 & 22.0 & 2.6 \\ 
   &  & 0.75 & 10.9 & 11.0 & 5.2 & 14.3 & 58.1 & 6.4 & 12.6 & 42.3 & 1.6 \\ 
   &  & 0.50 & 17.6 & 14.5 & 5.9 & 12.2 & 70.3 & 7.9 & 10.7 & 63.6 & 1.6 \\ 
   & 200 & 0.25 & 18.8 & 9.9 & 2.7 & 7.3 & 62.4 & 6.4 & 6.0 & 69.9 & 5.4 \\ 
   &  & 0.75 & 20.4 & 22.6 & 11.8 & 12.9 & 85.9 & 3.6 & 13.2 & 85.5 & 1.3 \\ 
   &  & 0.50 & 47.5 & 29.7 & 12.2 & 11.3 & 94.5 & 6.1 & 10.7 & 95.2 & 2.2 \\ 
  0.8 & 100 & 0.25 & 7.3 & 2.7 & 0.0 & 9.4 & 11.3 & 25.6 & 6.6 & 14.1 & 3.9 \\ 
   &  & 0.75 & 10.1 & 5.0 & 0.6 & 18.9 & 34.1 & 23.4 & 14.0 & 36.1 & 2.9 \\ 
   &  & 0.50 & 12.7 & 4.4 & 0.3 & 13.1 & 32.6 & 23.0 & 10.4 & 45.8 & 2.8 \\ 
   & 200 & 0.25 & 16.8 & 3.0 & 0.6 & 8.7 & 23.8 & 21.6 & 6.7 & 40.4 & 3.6 \\ 
   &  & 0.75 & 13.2 & 7.8 & 1.0 & 17.3 & 55.8 & 14.6 & 14.9 & 70.6 & 1.2 \\ 
   &  & 0.50 & 36.8 & 9.5 & 0.9 & 15.7 & 61.5 & 18.6 & 13.4 & 79.9 & 2.1 \\ 
   \hline
\end{tabular}
\endgroup
\end{table}

\begin{table}[t!]
\centering
\caption{Percentage of rejection of $H_0$ computed from 1000 samples of size $n \in \{100, 200\}$ such that the $\ip{nt}$ first observations are from a GEV(0,1,$\xi$) and the $n - \ip{nt}$ last observations are from a GEV(0.5,1,$\xi$).} 
\label{resH1mu_mu_0}
\begingroup\small
\begin{tabular}{rrrrrrarrarr}
  \hline
  \multicolumn{6}{c}{} & \multicolumn{3}{c}{PWM / $T_{g,n}$} & \multicolumn{3}{c}{GPWM / $S_{h,n}$} \\ \cmidrule(lr){7-9} \cmidrule(lr){10-12} $\xi$ & $n$ & $t$ & $F_n$ & $\bar x_n$ & $s_n$ & $\tilde g_\mu$ & $\tilde g_\sigma$ & $\tilde g_\xi$  & $\tilde h_\mu$ & $\tilde h_\sigma$ & $\tilde h_\xi$  \\ \hline
-0.4 & 100 & 0.25 & 39.3 & 39.2 & 2.2 & 27.4 & 3.9 & 3.9 & 27.6 & 1.1 & 0.5 \\ 
   &  & 0.75 & 36.5 & 38.4 & 3.9 & 27.9 & 3.4 & 3.8 & 19.3 & 1.0 & 0.5 \\ 
   &  & 0.50 & 58.0 & 61.9 & 2.1 & 48.0 & 3.5 & 4.1 & 42.7 & 0.6 & 0.5 \\ 
   & 200 & 0.25 & 67.6 & 72.1 & 3.6 & 57.7 & 3.2 & 4.5 & 53.9 & 1.6 & 1.5 \\ 
   &  & 0.75 & 68.6 & 73.3 & 4.1 & 57.1 & 5.1 & 4.2 & 47.8 & 2.5 & 1.0 \\ 
   &  & 0.50 & 88.6 & 91.0 & 1.2 & 79.3 & 3.0 & 3.7 & 74.1 & 1.4 & 2.0 \\ 
  0.0 & 100 & 0.25 & 28.2 & 19.7 & 3.7 & 26.5 & 3.4 & 2.6 & 23.8 & 0.9 & 1.1 \\ 
   &  & 0.75 & 28.2 & 23.3 & 2.2 & 29.0 & 4.7 & 4.2 & 21.2 & 1.8 & 2.0 \\ 
   &  & 0.50 & 48.6 & 39.0 & 1.9 & 48.3 & 2.3 & 3.4 & 42.1 & 0.7 & 1.7 \\ 
   & 200 & 0.25 & 52.8 & 42.2 & 4.0 & 54.6 & 4.1 & 2.5 & 52.8 & 2.3 & 1.8 \\ 
   &  & 0.75 & 51.7 & 42.4 & 1.7 & 54.9 & 4.5 & 4.0 & 48.1 & 3.6 & 2.7 \\ 
   &  & 0.50 & 76.7 & 66.6 & 3.5 & 78.1 & 3.7 & 3.1 & 76.1 & 2.4 & 2.3 \\ 
  0.4 & 100 & 0.25 & 23.1 & 4.7 & 0.5 & 22.1 & 4.7 & 4.7 & 19.8 & 2.0 & 2.3 \\ 
   &  & 0.75 & 22.8 & 7.8 & 0.5 & 27.2 & 5.2 & 4.3 & 22.0 & 2.8 & 3.1 \\ 
   &  & 0.50 & 43.5 & 10.6 & 0.3 & 45.2 & 3.8 & 3.6 & 42.0 & 1.4 & 3.0 \\ 
   & 200 & 0.25 & 48.8 & 8.5 & 1.2 & 48.4 & 5.0 & 3.2 & 51.9 & 3.0 & 2.7 \\ 
   &  & 0.75 & 48.6 & 8.5 & 1.0 & 54.1 & 5.6 & 2.8 & 49.9 & 3.9 & 4.6 \\ 
   &  & 0.50 & 74.0 & 16.8 & 0.9 & 71.6 & 4.7 & 2.3 & 76.6 & 2.8 & 3.1 \\ 
  0.8 & 100 & 0.25 & 22.9 & 1.2 & 0.2 & 21.7 & 7.1 & 21.0 & 18.5 & 4.9 & 2.5 \\ 
   &  & 0.75 & 24.2 & 2.5 & 0.2 & 27.9 & 6.7 & 19.7 & 26.5 & 5.3 & 4.0 \\ 
   &  & 0.50 & 46.2 & 2.9 & 0.2 & 38.9 & 6.8 & 23.2 & 41.1 & 4.9 & 1.3 \\ 
   & 200 & 0.25 & 48.6 & 2.0 & 0.4 & 40.0 & 7.9 & 15.4 & 45.6 & 6.9 & 2.2 \\ 
   &  & 0.75 & 52.0 & 3.3 & 0.4 & 46.4 & 7.9 & 17.5 & 51.8 & 6.9 & 3.9 \\ 
   &  & 0.50 & 77.5 & 3.4 & 0.4 & 64.7 & 7.8 & 13.6 & 71.6 & 5.1 & 2.0 \\ 
   \hline
\end{tabular}
\endgroup
\end{table}

In a last experiment, we investigated the power of the tests for samples generated from normal distributions with different variances or expectations. The results are reported in Tables~23 
and~24 
of the supplementary material. In the case of a change in variance (resp.\ expectation), the proposed tests designed to be sensitive to a change in scale (resp.\ location) appear overall slightly less powerful than the ``$s_n^2$'' (resp.\ ``$\bar x_n$'') tests.

\section{Illustrations}
\label{sec:illus}

The results of the finite-sample experiments partially reported in the previous section suggest that the overall behavior of the tests based on the statistics $T_{g,n}$ in~\eqref{eq:T_g_n} is better than that of the tests based on the statistics $S_{h,n}$ in~\eqref{eq:S_h_n} when the distributions of the block maxima have finite variances. When they have an infinite variance but a finite expectation, the tests based on the statistics $S_{h,n}$ should however clearly be preferred as those based on $T_{g,n}$ do not hold their level anymore. 

When dealing with climate data, the assumption of finite variances is reasonable for temperatures, precipitation or winds, which is why we focus on the tests based on the statistics $T_{g,n}$ in the forthcoming applications. 

To illustrate the pratical use of the three selected tests, we applied them on some of the environmental datasets available in the \textsf{R} packages {\tt evd} \citep{evd}, {\tt extRemes} \citep{extRemes}, {\tt ismev} \citep{ismev} and {\tt ClusterMax} \citep{ClusterMax}. The observations consist of block maxima of wind speeds, temperatures, precipitation or sea levels, and can reasonably be considered to be independent. The interest reader can easily plot the data after installing the aforementioned \textsf{R} packages.

The main practical difficulty that prevented a direct application of the tests is that all the datasets contain ties due to a limited precision of the underlying measurement instruments. To address this issue, we proceeded as follows. For every dataset, we computed the smallest absolute difference $d$ between any two distinct observations. Assuming that $d$ is a good estimate of the precision of the measurements, we added a random component uniformly distributed over $(0,d)$ to every observation. For every dataset, 1000 different de-tied samples were generated with the hope that at least one of these is close to the true unobserved sample, and the tests were applied to each such de-tied sample. The minimum and the maximum of the obtained p-values are reported in Table~\ref{illus}. The minimum and the maximum of the estimates $\tilde g_\mu(\hat{\bm b}_{\bm \nu,1:n})$, $\tilde g_\sigma(\hat{\bm b}_{\bm \nu,1:n})$ and $\tilde g_\xi(\hat{\bm b}_{\bm \nu,1:n})$ are provided as well. Under $H_0$ in~\eqref{eq:H0} with $F$ the c.d.f.\ of the GEV, the three latter quantities are estimates of the corresponding location, scale and shape parameters obtained using the PWM method summarized in Section~\ref{sec:PWM_method}.

\begin{sidewaystable}[t!]
\centering
\caption{Minimum and maximum of $\tilde g_\mu(\hat{\bm b}_{1:n})$, $\tilde g_\sigma(\hat{\bm b}_{1:n})$ and $\tilde g_\xi(\hat{\bm b}_{1:n})$ as well as minimum and maximum of the p-values of the tests based on $T_{g,n}$ in~\eqref{eq:T_g_n} with $g$ respectively equal to $\tilde g_\mu$, $\tilde g_\sigma$, $\tilde g_\xi$ for 1000 different de-tied samples generated from each dataset as explained in Section~\ref{sec:illus}. The first two columns provide the names of the \textsf{R} packages containing the datasets and the names of the latter. The columns $n$ and $n'$ report the number of block maxima and the number of distinct block maxima, respectively.} 
\label{illus}
{\small
\begin{tabular}{llrrrrrrrrrrrrrr}
  \hline
  \multicolumn{4}{c}{} & \multicolumn{6}{c}{Estimates} & \multicolumn{6}{c}{p-values of the test based on $T_{g,n}$ with} \\
\cmidrule(lr){5-10} \cmidrule(lr){11-16}  \multicolumn{4}{c}{} & \multicolumn{2}{c}{$\tilde g_\mu(\hat{\bm \beta}_{\bm \nu,1:n})$} & \multicolumn{2}{c}{$\tilde g_\sigma(\hat{\bm \beta}_{\bm \nu,1:n})$} & \multicolumn{2}{c}{$\tilde g_\xi(\hat{\bm \beta}_{\bm \nu,1:n})$} & \multicolumn{2}{c}{$g=\tilde g_\mu$} & \multicolumn{2}{c}{$g=\tilde g_\sigma$} & \multicolumn{2}{c}{$g=\tilde g_\xi$} \\
\cmidrule(lr){5-6} \cmidrule(lr){7-8} \cmidrule(lr){9-10} \cmidrule(lr){11-12} \cmidrule(lr){13-14} \cmidrule(lr){15-16} Package & Dataset & $n$ & $n'$ & min & max & min & max & min & max & min & max & min & max & min & max \\ \hline
evd & lisbon & 30 & 21 & 95.79 & 96.22 & 12.62 & 13.07 & -0.16 & -0.13 & 0.152 & 0.205 & 0.167 & 0.271 & 0.416 & 0.630 \\ 
  evd & oxford & 80 & 19 & 84.23 & 84.46 & 4.20 & 4.44 & -0.34 & -0.26 & 0.099 & 0.248 & 0.534 & 1.000 & 0.413 & 1.000 \\ 
  extRemes & HEAT (Tmax) & 43 & 9 & 112.89 & 113.22 & 1.91 & 2.32 & -0.41 & -0.21 & 0.002 & 0.029 & 0.590 & 1.000 & 0.766 & 1.000 \\ 
  extRemes & HEAT (-Tmin) & 43 & 15 & -70.70 & -70.37 & 3.67 & 4.07 & -0.18 & -0.09 & 0.000 & 0.002 & 0.130 & 0.529 & 0.173 & 0.898 \\ 
  extRemes & ftcanmax & 100 & 78 & 135.75 & 135.96 & 55.51 & 55.79 & 0.13 & 0.13 & 0.724 & 0.757 & 0.430 & 0.467 & 1.000 & 1.000 \\ 
  ismev & fremantle & 86 & 31 & 1.49 & 1.49 & 0.14 & 0.14 & -0.21 & -0.19 & 0.006 & 0.009 & 0.433 & 0.578 & 1.000 & 1.000 \\ 
  ismev & portpirie & 65 & 42 & 3.88 & 3.88 & 0.20 & 0.20 & -0.06 & -0.04 & 0.537 & 0.603 & 0.788 & 0.949 & 0.782 & 0.928 \\ 
  ClusterMax & Station 8 (St-Lizier) & 228 & 58 & 1.97 & 1.98 & 1.95 & 1.97 & 0.14 & 0.15 & 0.217 & 0.261 & 0.296 & 0.378 & 0.068 & 0.091 \\ 
  ClusterMax & Station 14 (Aurillac) & 228 & 73 & 2.39 & 2.40 & 2.87 & 2.89 & 0.29 & 0.29 & 0.152 & 0.174 & 0.622 & 0.683 & 0.124 & 0.138 \\ 
  ClusterMax & Station 57 (Nevers) & 228 & 56 & 1.82 & 1.84 & 1.90 & 1.92 & 0.18 & 0.18 & 0.287 & 0.349 & 0.155 & 0.202 & 0.034 & 0.046 \\ 
  ClusterMax & Station 78 (Niort) & 228 & 59 & 1.96 & 1.97 & 2.15 & 2.17 & 0.22 & 0.22 & 0.038 & 0.047 & 0.016 & 0.022 & 0.448 & 0.523 \\ 
   \hline
\end{tabular}
}
\end{sidewaystable}

Small maximal p-values provide evidence against $H_0$ in~\eqref{eq:H0}. Under the additional assumption of GEV distributed block maxima, the latter can be interpreted as evidence of change in the location, scale or shape parameter. Ignoring a necessary adjustment of the p-values due to multiple testing, the following conclusions could for instance be drawn from Table~\ref{illus}:
\begin{itemize}
\item there is very strong evidence of a change in the location parameter for the annual minimal temperature in the {\tt HEAT} dataset, and strong evidence of the same nature for the annual maximal temperature in the {\tt HEAT} dataset and for the annual maximal sea levels in the {\tt fremantle} dataset,
\item there is weak (resp.\ some) evidence of change in the shape parameter for the weekly maxima of hourly precipitation measured at Saint-Lizier (resp.\ Nevers), France, during the Fall season between 1993 and 2011,
\item there is strong evidence of a change in the scale parameter and some evidence of a change in the location parameter for the weekly maxima of hourly precipitation measured at Niort, France, during the Fall season between 1993 and 2011.

\end{itemize}

Notice that the tests, as implemented in \textsf{R} package {\tt npcp}, are fast. Obtaining the results reported in Table~\ref{illus} took less than 30 seconds on one 2.7 GHz processor. 

Let us end this section with a practical remark. To test $H_0$ in~\eqref{eq:H0} from a single sample of block maxima, three tests, based on the statistics $T_{g,n}$ in~\eqref{eq:T_g_n} or on the statistics $S_{h,n}$ in~\eqref{eq:S_h_n}, are to be used. As noted by a referee, such an approach, although interesting from a data analysis perspective, lacks a simple decision rule regarding the rejection or not of $H_0$. The later could be the subject of future research. A simple yet conservative approach would be to use a Bonferroni correction and reject $H_0$ as soon one p-value is below a third of the chosen significance level.


\paragraph{Acknowledgments}

The authors would like to thank two anonymous Referees and an Associate Editor for their constructive comments on an earlier version of this manuscript. Part of Philippe Naveau’s work has been supported by the ANR-DADA, LEFE-INSU-Multirisk, AMERISKA, A2C2 and Extremoscope projects. Part of the research was done while the second author was visiting the IMAGE-NCAR group in Boulder, CO, USA.

\appendix

\section{Proof of Proposition~\ref{prop:beta_ineqs}}
\label{proof:prop:beta_ineqs}

\begin{proof}
Some thought reveals that
$$
\beta_i = \Ex[X \{ F(X) \}^{i-1} ] = i^{-1} \Ex\{ \max(X_1,\dots,X_i) \}, \qquad i \in \{1,2,3\},
$$
where $X_1,\dots,X_i$ are independent copies of $X$. 

The first inequality is obtained by noticing that $\max(X_1,X_2) > (X_1+X_2)/2$ almost surely since $\Pr(X_1 = X_2) = 0$ by continuity of $F$, and by taking expectations on both sides. The second inequality follows from the fact that $3 \max(X_1,X_2,X_3) > \max(X_1,X_2) + \max(X_1,X_3) + \max(X_2,X_3)$ almost surely. The last inequality is a consequence of the fact that
$$
2 \max(X_1,X_2) + 2 \max(X_1,X_3) + 2 \max(X_2,X_3) > 3 \max(X_1,X_2,X_3) + X_1+X_2+X_3
$$
almost surely.
\end{proof}

\section{Weak convergence of the sequential weighted uniform empirical process}
\label{sec:weighted_uniform}

The purpose of this appendix is to derive the weak convergence of a sequential version of the weighted uniform empirical process. The latter result, that may be of independent interest, is used in the proof of Proposition~\ref{prop:weak_Mn}, itself necessary for the proof, given in Appendix~\ref{proof:prop:weak_S_g_n}, of Proposition~\ref{prop:weak_S_g_n}. The forthcoming developments rely heavily on Appendix~G of \cite{GenSeg09}. 

Let $U_1,\dots,U_n$ be an i.i.d.\ sample from the standard uniform distribution. For any integers $1 \leq k \leq l \leq n$, let $G_{k:l}$ be the empirical c.d.f.\ computed from the sample $U_k,\dots,U_l$ with the convention that $G_{k:l} = 0$, for all $k>l$. Then, define the {\em two-sided sequential uniform empirical process} by
\begin{equation}
\label{eq:Bn}
\B_n(s,t,u) = \frac{1}{\sqrt{n}} \sum_{i=\ip{ns}+1}^{\ip{nt}} \{ \1(U_i \leq u) - u \} = \sqrt{n} \lambda_n(s,t) \{G_{\ip{ns}+1:\ip{nt}}(u) - u \},
\end{equation}
for $(s,t,u) \in \Delta \times [0,1]$, where $\Delta = \{ (s,t) \in [0,1]^2 : s \leq t\}$ and $\lambda_n(s,t) = (\ip{nt}-\ip{ns})/n$, $(s,t) \in \Delta$. The process $\B_n(0,\cdot,\cdot)$ is the standard sequential uniform empirical process and it is well-known \citep[see, e.g.,][Theorem 2.12.1]{vanWel96} that $\B_n(0,\cdot,\cdot) \leadsto \B(0,\cdot,\cdot)$ in $\ell^\infty([0,1]^2)$, where $\B(0,\cdot,\cdot)$ is a centered Gaussian process with continuous trajectories, sometimes called a {\em Kieffer--Müller process}, whose covariance function is given by
$$
\Cov\{ \B(0,s,u), \B(0,t,v) \} = (s \wedge t) ( u \wedge v - uv), \qquad (s,t,u,v) \in [0,1]^4.
$$
Since, for any $(s,t) \in \Delta$, $\B_n(s,t,\cdot) = \B_n(0,t,\cdot) - \B_n(0,s,\cdot)$, we immediately obtain from the continuous mapping that $\B_n \leadsto \B$ in $\ell^\infty(\Delta \times [0,1])$, where, for any $(s,t) \in \Delta$, $\B(s,t,\cdot) = \B(0,t,\cdot) - \B(0,s,\cdot)$.

Analogously to \cite{GenSeg09}, let $\alpha \geq 0$ and consider a weighted version of the processes $\B_n$ defined as
\begin{equation}
\label{eq:W_alpha_n}
\W_{\alpha,n}(s,t,u) = 
\left\{
\begin{array}{ll}
\disp \frac{\B_n(s,t,u)}{ \{ u (1-u) \}^{\alpha}}, &\mbox{if } u \in (0,1), \\
0, &\mbox{if } u \in \{0,1\}. 
\end{array}
\right.
\end{equation}
The following result is then a univariate sequential extension of Theorem~G.1 of \cite{GenSeg09}. 

\begin{prop}
\label{prop:weak_W_alpha_n}
For any $\alpha \in [0,1/2)$, $\W_{\alpha,n} \leadsto \W_\alpha$ in $\ell^\infty(\Delta \times [0,1])$, where $\W_{\alpha}$ is a centered Gaussian process with continuous trajectories such that $\W_{\alpha}(0,\cdot,u) = 0$ for $u \in \{0,1\}$, 
$$
\Cov\{ \W_\alpha(0,s,u), \W_\alpha(0,t,v) \} = \frac{(s \wedge t) ( u \wedge v - uv)}{\{ u (1-u) v (1-v) \}^{\alpha}}, \qquad (s,t,u,v) \in [0,1]^2 \times (0,1)^2,
$$
and $\W_{\alpha}(s,t,\cdot) = \W_{\alpha}(0,t,\cdot) - \W_{\alpha}(0,s,\cdot)$ for all $(s,t) \in \Delta$.
\end{prop} 

\begin{proof}
Fix $\alpha \in [0,1/2)$. Analogously to~\citet[Appendix~G]{GenSeg09}, for any $u \in (0,1)$, let $f_u:(0,1) \to \R$ be defined as $f_u(x) = \{u (1-u) \}^{-\alpha} \{ \1(x \leq u) - u \}$, $x \in (0,1)$, and let $\FF = \{ f_u : u \in (0,1) \} \cup \{ 0 \}$, where $0$ is the function vanishing everywhere on~$(0,1)$. Furthermore, let $P$ be the uniform probability distribution on~$(0,1)$. Lemma~G.2 of \citet{GenSeg09} then states that the collection $\FF$ is $P$-Donsker \citep[see, e.g.,][]{vanWel96,Kos08}, which is equivalent to saying that there exists a $P$-Brownian bridge $\G$ such that $\G_n \leadsto \G$ in $\ell^\infty(\FF)$, where, for any $f \in \FF$, $\G_n(f)= \sqrt{n} (\Pb_n f - Pf )$ with 
$$
Pf = \int_{[0,1]} f(x) \dd x, \qquad \Pb_n f = \frac{1}{n} \sum_{i=1}^n f(U_i),
$$
and the convention that $\Pb_0 f = 0$. For any $s \in [0,1]$ and $f \in \FF$, let $\Z_n(s,f) = \sqrt{n} \lambda_n(0,s) (\Pb_{\ip{ns}} f - Pf ) = \sqrt{n} \lambda_n(0,s) \Pb_{\ip{ns}} f$, where the second equality follows from the fact that $Pf=0$ for all $f \in \FF$. An immediate consequence of Theorem 2.12.1 of \cite{vanWel96} is then that $\Z_n \leadsto \Z$ in $\ell^\infty([0,1] \times \FF)$, where $\Z$ is a tight centered Gaussian process with covariance function 
\begin{equation}
\label{eq:covZ}
\Cov\{\Z(s,f),\Z(t,g)\} = (s \wedge t) (Pfg - Pf Pg) = (s \wedge t) Pfg, \qquad s,t \in [0,1], f,g \in \FF.
\end{equation}
To conclude the proof, we proceed analogously to \citet[Appendix~G]{GenSeg09}. Let $\phi:[0,1] \to \FF$ be defined as
$$
\phi(u) = \left\{
\begin{array}{ll}
f_u, & \mbox{if } u \in (0,1), \\
0, & \mbox{if } u \in \{0,1\}.
\end{array}
\right.
$$
Then, let $T:\ell^\infty([0,1] \times \FF) \to \ell^\infty([0,1]^2)$ be defined as $T(z)(s,u) = z(s,\phi(u))$ for all $z \in \ell^\infty([0,1] \times \FF)$ and $(s,u) \in [0,1]^2$, and notice that $T(\Z_n) = \W_{\alpha,n}(0,\cdot,\cdot)$. Indeed, for any $s \in [0,1]$, if $u \in (0,1)$,
$$
T(\Z_n)(s,u) = \Z_n(s,f_u) = \sqrt{n} \lambda_n(0,s) \Pb_{\ip{ns}} f_u = \W_{\alpha,n}(0,s,u),
$$
and, if $u \in \{0,1\}$,
$$
T(\Z_n)(s,u) = \Z_n(s,0) = \sqrt{n} \lambda_n(0,s) \Pb_{\ip{ns}} 0 = \W_{\alpha,n}(0,s,u).
$$
The map $T$ being linear and bounded, it is continuous. Hence, by the continuous mapping theorem, $T(\Z_n) \leadsto T(\Z)$ in $\ell^\infty([0,1]^2)$. Let $\W_{\alpha}(0,\cdot,\cdot) = T(\Z)$. Given the definition of $T$, since $\Z$ is centered, so is $\W_{\alpha}(0,\cdot,\cdot)$, and since the finite-dimensional distributions of $\Z$ are Gaussian, so are those of $\W_{\alpha}(0,\cdot,\cdot)$. Furthermore, the covariance function of $\W_{\alpha}(0,\cdot,\cdot)$ follows immediately from~\eqref{eq:covZ}. It remains to verify that the trajectories of $\W_{\alpha}(0,\cdot,\cdot)$ are continuous. From \citet[Chapter 2.12]{vanWel96}, we know that the trajectories of $\Z$ are almost surely uniformly continuous with respect to the natural semimetric $|s-t| + \rho(f,g)$ on $[0,1] \times \FF$, where 
$$
\rho^2(f,g) = \Ex[ \{\G(f) - \G(g)\}^2 ], \qquad f,g \in \FF.
$$ 
From \citet[page 3020]{GenSeg09}, we have that the map $\phi$ is continuous with respect to the Euclidean metric on $[0,1]$ and the standard deviation semimetric $\rho$ on $\FF$. As a consequence, the trajectories of $(s,u) \mapsto \W_{\alpha}(0,s,u) = \Z(s,\phi(u))$ are almost surely continuous on $[0,1]^2$. 

Finally, since $\W_{\alpha,n}(s,t,\cdot) = \W_{\alpha,n}(0,t,\cdot) - \W_{\alpha,n}(0,s,\cdot)$ for all $(s,t) \in \Delta$, we immediately obtain from the continuous mapping theorem that $\W_{\alpha,n} \leadsto \W_\alpha$ in $\ell^\infty(\Delta \times [0,1])$, where, for any $(s,t) \in \Delta$, $\W_{\alpha}(s,t,\cdot) = \W_{\alpha}(0,t,\cdot) - \W_{\alpha}(0,s,\cdot)$.
\end{proof}


\section{Weak convergence of sequential PWM processes}
\label{proof:prop:weak_Mn}

In order to study the limiting null distribution of the test statistics $S_{g,n}$ in~\eqref{eq:S_g_n} and $T_{g,n}$ in~\eqref{eq:T_g_n}, it is necessary to investigate the asymptotics of sequential empirical processes constructed from PWM estimators. The main result of this section is used in the proof, given in Appendix~\ref{proof:prop:weak_S_g_n}, of Proposition~\ref{prop:weak_S_g_n}, but it might also be of independent interest. 

As we consider a slightly more general theoretical framework than in the main sections of the paper, a richer notation is necessary. Let $X$ be a random variable with c.d.f.\ $F$ and let $\nu_1,\dots,\nu_p$ be $p$ functions on $[0,1]$. The reals $\beta_{\nu} = \Ex [ X \nu\{F(X)\} ]$, $\nu \in \{\nu_1,\dots,\nu_p\}$, assuming they exist, are called {\em generalized} PWM following \cite{DieGuiNavRib08}. When the functions $\nu_i$ are defined as $\nu_i(x) = x^{r_i} (1 - x)^{s_i}$, $x \in [0,1]$, for some integers $r_i \geq 0$ and $s_i \geq 0$, one recovers a very usefull subset of ``classical'' PWM that can be used to fit a large number of distributions \citep[see, e.g.,][]{GreLanWal79,HosWalWoo85,HosWal87}. In the sequel, we write~$\bm \nu = (\nu_1,\dots,\nu_p)$ and $\bm \beta_{\bm \nu} = (\beta_{\nu_1},\dots,\beta_{\nu_p})$.

Assuming that we have at hand a random sample $X_1,\dots,X_n$ from an unknown c.d.f.~$F$, a natural nonparametric estimator of $\bm \beta_{\bm \nu}$ based on a subsample $X_k,\dots,X_l$, $1 \leq k \leq l \leq n$, is 
\begin{equation}
\label{eq:hatbbetakl}
\hat{\bm \beta}_{\bm \nu,k:l} = (\hat \beta_{\nu_1,k:l},\dots,\hat \beta_{\nu_p,k:l}),
\end{equation}
where $\hat \beta_{\nu_i,k:l}$ stands for $\hat \beta_{i,k:l}$ defined in~\eqref{eq:hatbetakl} with the difference that $\gamma$ in~\eqref{eq:Fkl} is set to zero. 

For any $x \geq 0$, let $\ip{x}$ be the largest integer smaller or equal than $x$. Furthermore, recall that $\Delta = \{ (s,t) \in [0,1]^2 : s \leq t\}$ and that $\lambda_n(s,t) = (\ip{nt}-\ip{ns})/n$, $(s,t) \in \Delta$. The aim of this section is to establish the weak convergence  of the vector of empirical processes $\Mb_{\bm \nu,n} = \Big( \M_{\nu_1,n},\dots,\M_{\nu_p,n} \Big)$, where, for any $\nu \in \{\nu_1,\dots,\nu_p\}$,
\begin{equation}
\label{eq:M_nu_n}
\M_{\nu,n}(s,t) = \sqrt{n} \lambda_n(s,t)  ( \hat \beta_{\nu,\ip{ns}+1:\ip{nt}} - \beta_{\nu} ), \qquad (s,t) \in \Delta.
\end{equation}

The desired result is established under the following condition.

\begin{cond}
\label{cond:main}
The functions $\nu_1,\dots,\nu_p$ are continuously differentiable on $[0,1]$. Furthermore, $X_1,\dots,X_n$ are i.i.d.\ from a continuous distribution with c.d.f.\ $F$ such that:
\begin{enumerate}[(i)]
\item $\beta_{\nu_i} = \Ex [ X_1 \nu_i\{F(X_1)\} ] < \infty$ for all $i \in \{1,\dots,p\}$,
\item $\Ex( [  X_1 \nu_i\{F(X_1)\} - \beta_{\nu_i} ][  X_1 \nu_j\{F(X_1)\} - \beta_{\nu_j} ] ) < \infty$ for all $i,j \in \{1,\dots,p\}$,
\item there exists $\alpha \in [0,1/2)$ such that $\sup_{x \in \R} | H_\alpha(x) | < \infty$, where
\begin{equation*}
H_\alpha(x) = x [ F(x) \{ 1-F(x)\} ]^\alpha, \qquad x \in \R.
\end{equation*}
\end{enumerate}
\end{cond}

It is easy to verify that if one of the functions $\nu_1,\dots,\nu_p$ is constant on $[0,1]$, then~(i) and~(ii) in Condition~\ref{cond:main} are equivalent to requiring that $\Ex(X_1^2) < \infty$. 

\begin{prop}[Limiting null distribution of $\Mb_{\bm \nu,n}$]
\label{prop:weak_Mn}
Under Condition~\ref{cond:main},
$$
\Mb_{\bm \nu,n} = \Big( \M_{\nu_1,n},\dots,\M_{\nu_p,n} \Big) \leadsto \Mb_{\bm \nu} = \Big( \M_{\nu_1},\dots,\M_{\nu_p} \Big)
$$
in $\{ \ell^\infty(\Delta) \}^p$, where $\M_{\nu_1,n},\dots,\M_{\nu_p,n}$ are defined in~\eqref{eq:M_nu_n} and $\M_{\nu_1},\dots,\M_{\nu_p}$ are centered Gaussian processes with continuous trajectories such that
\begin{equation}
\label{eq:M_nu_st}
\M_{\nu_i}(s,t) = \M_{\nu_i}(0,t) - \M_{\nu_i}(0,s), \qquad (s,t) \in \Delta, \, i \in  \{1,\dots,p\},
\end{equation}
and
\begin{equation}
\label{eq:covW}
\Cov\{ \M_{\nu_i}(0,s), \M_{\nu_j}(0,t) \} = (s \wedge t) \Cov(Y_{\nu_i},Y_{\nu_j}), \qquad (s,t) \in \Delta, \, i,j \in  \{1,\dots,p\},
\end{equation}
where, for any $\nu \in \{\nu_1,\dots,\nu_p\}$, $Y_\nu$ is defined in~\eqref{eq:Ynu}.
\end{prop}

\begin{proof}
Fix $\nu \in \{\nu_1,\dots,\nu_p\}$ and consider the decomposition 
$$
\M_{\nu,n}(s,t) = \X_{\nu,n}(s,t) + \Y_{\nu,n}(s,t), \qquad (s,t) \in \Delta,
$$
where
\begin{align}
\label{eq:X_nu_n}
\X_{\nu,n}(s,t) &= \sqrt{n} \lambda_n(s,t)  ( \hat \beta_{\nu,\ip{ns}+1:\ip{nt}} - \tilde \beta_{\nu,\ip{ns}+1:\ip{nt}} ), \\
\nonumber 
\Y_{\nu,n}(s,t) &= \sqrt{n} \lambda_n(s,t)  ( \tilde \beta_{\nu,\ip{ns}+1:\ip{nt}} - \beta_{\nu} ), 
\end{align}
with 
\begin{equation}
\label{eq:tilde_beta_st}
\tilde \beta_{\nu,\ip{ns}+1:\ip{nt}} = \frac{1}{\ip{nt} - \ip{ns}}  \sum_{i=\ip{ns}+1}^{\ip{nt}} X_i \nu\{F(X_i)\}.
\end{equation}
In addition, let $U_1,\dots,U_n$ be the unobservable sample obtained from $X_1,\dots,X_n$ by the probability integral transformations $U_i = F(X_i)$, $i \in \{1,\dots,n\}$, and recall the definition of the empirical process $\W_{\alpha,n}$ in~\eqref{eq:W_alpha_n} and its weak limit $\W_\alpha$ given in Proposition~\ref{prop:weak_W_alpha_n}. The first step of the proof is to show that 
\begin{equation}
\label{eq:step_join}
(\Y_{\nu_1,n},\dots,\Y_{\nu_p,n},\W_{\alpha,n}) \leadsto (\Y_{\nu_1},\dots,\Y_{\nu_p},\W_\alpha)
\end{equation}
 in $\{ \ell^\infty(\Delta) \}^p \times \ell^\infty(\Delta \times [0,1])$.

The proof of the convergence in distribution of the finite-dimensional distributions of $(\Y_{\nu_1,n},\dots,\Y_{\nu_p,n},\W_{\alpha,n})$ to those of $(\Y_{\nu_1},\dots,\Y_{\nu_p},\W_\alpha)$ 
is omitted for the sake of brevity. To prove~\eqref{eq:step_join}, it remains thus to show joint asymptotic tightness, which is implied by marginal asymptotic tightness. From Proposition~\ref{prop:weak_W_alpha_n}, we have that $\W_{\alpha,n} \leadsto \W_\alpha$ in $\ell^\infty(\Delta \times [0,1])$. Next, fix $\nu \in \{\nu_1,\dots,\nu_p\}$ and notice that, since $\sigma_{\nu}^2 = \Var[X_1 \nu \{F(X_1)\}] < \infty$, we have, from the functional central limit theorem, that $\Y_{\nu,n}(0,\cdot) \leadsto \Y_{\nu}(0,\cdot)$ in $\ell^\infty([0,1])$, where $\Y_{\nu}(0,\cdot) = \sigma_{\nu} \U$ with $\U$ a standard Brownian motion. Since, for any $(s,t) \in \Delta$, $\Y_{\nu,n}(s,t) = \Y_{\nu,n}(0,t) - \Y_{\nu,n}(0,s)$, an immediate consequence of the continuous mapping theorem is that $\Y_{\nu,n} \leadsto \Y_{\nu}$ in $\ell^\infty(\Delta)$, where, for any $(s,t) \in \Delta$, $\Y_{\nu}(s,t) = \sigma_{\nu} \{ \U(t) - \U(s) \}$. This completes the proof of~\eqref{eq:step_join}.

Fix again $\nu \in \{\nu_1,\dots,\nu_p\}$. Since $\sup_{x \in \R} | H_\alpha(x) | < \infty$, $\sup_{u \in [0,1]} | \nu'(u) | < \infty$ and $F$ is continuous, the map $\psi_{\nu,\alpha}:\ell^\infty(\Delta \times [0,1]) \to \ell^\infty(\Delta)$ defined as 
\begin{equation}
\label{eq:psi_nu_alpha}
\psi_{\nu,\alpha}(f)(s,t) = \int_\R \nu'\{F(x)\} H_\alpha(x) f\{s,t,F(x) \} \dd F(x), \quad (s,t) \in \Delta, \quad f \in \ell^\infty(\Delta \times [0,1]), 
\end{equation}
is continuous. The weak convergence in~\eqref{eq:step_join} and the continuous mapping theorem then imply that 
\begin{multline*}
\Big(\Y_{\nu_1,n} + \psi_{\nu_1,\alpha}(\W_{\alpha,n}),\dots, \Y_{\nu_p,n} + \psi_{\nu_p,\alpha}(\W_{\alpha,n}) \Big) \leadsto \Big(\Y_{\nu_1} + \psi_{\nu_1,\alpha}(\W_{\alpha}),\dots, \Y_{\nu_p} + \psi_{\nu_p,\alpha}(\W_{\alpha}) \Big)
\end{multline*}
in~$\{\ell^\infty(\Delta)\}^p$. To complete the proof, it remains thus to show that, for any $\nu \in \{\nu_1,\dots,\nu_p\}$, 
\begin{equation}
\label{eq:ae_X_phi}
\sup_{(s,t) \in \Delta} | \X_{\nu,n}(s,t) -  \psi_{\nu,\alpha}(\W_{\alpha,n})(s,t)\} | = o_\Pr(1),
\end{equation}
where $\X_{\nu,n}$ is defined in~\eqref{eq:X_nu_n}. Fix $\nu \in \{\nu_1,\dots,\nu_p\}$ and consider the empirical process
\begin{equation}
\label{eq:X_nu_n'}
\X_{\nu,n}'(s,t) = \int_\R  \nu'\{F(x)\} H_\alpha(x) \W_{\alpha,n}\{s,t,F(x)\}  \dd F_{\ip{ns}+1:\ip{nt}}(x), \qquad (s,t) \in \Delta,
\end{equation}
where $F_{\ip{ns}+1:\ip{nt}}$ is defined in~\eqref{eq:Fkl} with $\gamma=0$. To prove~\eqref{eq:ae_X_phi}, we shall first show that 
\begin{equation}
\label{eq:step1}
\sup_{(s,t) \in \Delta} |\X_{\nu,n}(s,t) - \X'_{\nu,n}(s,t) | = o_\Pr(1),
\end{equation}
and then that
\begin{equation}
\label{eq:step2}
\sup_{(s,t) \in \Delta} |\X'_{\nu,n}(s,t) - \psi_{\nu,\alpha}(\W_{\alpha,n})(s,t) | = o_\Pr(1).
\end{equation}

\noindent
{\em Proof of~\eqref{eq:step1}.}
On one hand, starting from~\eqref{eq:X_nu_n}, we have that, for any $(s,t) \in \Delta$,
\begin{align*}
\X_{\nu,n}(s,t) &= \frac{1}{\sqrt{n}} \sum_{i=\ip{ns}+1}^{\ip{nt}} X_i \left[ \nu\{F_{\ip{ns}+1:\ip{nt}}(X_i)\} - \nu\{F(X_i)\} \right] \\
&= \int_\R x \sqrt{n} \lambda_n(s,t) \left[ \nu\{F_{\ip{ns}+1:\ip{nt}}(x)\} - \nu\{F(x) \} \right]  \dd F_{\ip{ns}+1:\ip{nt}}(x).
\end{align*}
On the other hand,~\eqref{eq:X_nu_n'} can be rewritten as 
$$
\X'_{\nu,n}(s,t) = \int_\R x \nu'\{F(x)\} \sqrt{n} \lambda_n(s,t) \{ F_{\ip{ns}+1:\ip{nt}}(x) - F(x) \} \dd F_{\ip{ns}+1:\ip{nt}}(x), \quad (s,t) \in \Delta.
$$
Then, the supremum on the left of~\eqref{eq:step1} is smaller than
\begin{multline}
\label{eq:supdiff}
\sup_{(s,t) \in \Delta} \sqrt{n} \lambda_n(s,t) \int_\R |x|  \left| \nu\{F_{\ip{ns}+1:\ip{nt}}(x)\} - \nu\{F(x)\} \right.  \\ \left. - \nu'\{F(x)\}  \{ F_{\ip{ns}+1:\ip{nt}}(x) - F(x) \} \right|  \dd F_{\ip{ns}+1:\ip{nt}}(x).
\end{multline}
By the mean value theorem, for any $x \in \R$, there exists $U^*_{n,s,t,x} \in (0, 1)$ such that
\begin{multline*}
\nu\{F_{\ip{ns}+1:\ip{nt}}(x)\} - \nu\{F(x)\} = \nu'[F(x) + U^*_{n,s,t,x} \{ F_{\ip{ns}+1:\ip{nt}}(x) - F(x) \} ] \\ \times \{ F_{\ip{ns}+1:\ip{nt}}(x) - F(x) \}.
\end{multline*}
Hence,~\eqref{eq:supdiff} is smaller than
\begin{multline*}
\sup_{(s,t) \in \Delta} \int_\R |x|  \left| \nu'[F(x) + U^*_{n,s,t,x} \{ F_{\ip{ns}+1:\ip{nt}}(x) - F(x) \}] - \nu'\{F(x)\} \right| \\ \times | \sqrt{n} \lambda_n(s,t) \{ F_{\ip{ns}+1:\ip{nt}}(x) - F(x) \} |  \dd F_{\ip{ns}+1:\ip{nt}}(x), 
\end{multline*}
that is, 
\begin{multline*}
\sup_{(s,t) \in \Delta} \int_\R |H_\alpha(x)|  \left| \nu'[F(x) + U^*_{n,s,t,x} \{ F_{\ip{ns}+1:\ip{nt}}(x) - F(x) \}] - \nu'\{F(x)\} \right| \\ \times | \W_{\alpha,n}\{s,t,F(x)\} |  \dd F_{\ip{ns}+1:\ip{nt}}(x).
\end{multline*}
Therefore,
\begin{multline}
\label{eq:supdiff2}
\sup_{(s,t) \in \Delta} |\X_{\nu,n}(s,t) - \X'_{\nu,n}(s,t) | \\ \leq \sup_{(s,t,x) \in \Delta \times \R} \left| \nu'[F(x) + U^*_{n,s,t,x} \{ F_{\ip{ns}+1:\ip{nt}}(x) - F(x) \}] - \nu'\{F(x)\} \right| \\ \times \sup_{(s,t,x) \in \Delta \times \R} | \W_{\alpha,n}\{s,t,F(x)\} | \times \sup_{x \in \R} |H_\alpha(x)|. 
\end{multline}
Fix $\eps, \eta > 0$. From the previous inequality, since $\sup_{u \in [0,1]} |\nu'(u)| < \infty$, $\sup_{x \in \R} | H_\alpha(x) | < \infty$, and using the fact that $\W_{\alpha,n}(s,t,\cdot) = 0$ when $s=t$ and the asymptotic uniform equicontinuity in probability of $\W_{\alpha,n}$ (a consequence of Proposition~\ref{prop:weak_W_alpha_n}), there exists $\delta \in (0,1)$ such that, for all sufficiently large~$n$, 
\begin{multline*}
\Pr \left( \sup_{(s,t) \in \Delta \atop t - s < \delta} |\X_{\nu,n}(s,t) - \X'_{\nu,n}(s,t) | > \eps \right) \\ \leq \Pr \left( 2 \sup_{u \in [0,1]} |\nu'(u)| \times \sup_{x \in \R} | H_\alpha(x) | \times \sup_{(s,t,u) \in \Delta \times [0,1] \atop t - s < \delta} | \W_{\alpha,n}(s,t,u) | > \eps \right) < \eta/2.
\end{multline*}
To prove~\eqref{eq:step1}, it remains therefore to show that, for all sufficiently large~$n$,
$$
\Pr \left( \sup_{(s,t) \in \Delta \atop t - s \geq \delta} |\X_{\nu,n}(s,t) - \X'_{\nu,n}(s,t) | > \eps \right) < \eta/2.
$$
Let $\Delta^\delta = \{(s,t) \in \Delta : t - s \geq \delta\}$. To show the latter, it suffices to show that 
$$
\sup_{(s,t) \in \Delta^\delta} |\X_{\nu,n}(s,t) - \X'_{\nu,n}(s,t) | = o_\Pr(1).
$$
The fact that $\B_n \leadsto \B$ in $\ell^\infty(\Delta \times [0,1])$, where $\B_n$ is defined in~\eqref{eq:Bn}, implies that 
\begin{equation}
\label{eq:unifconv}
\sup_{(s,t,u) \in \Delta^\delta \times [0,1]} | G_{\ip{ns}+1:\ip{nt}}(u) - u| = \sup_{(s,t,x) \in \Delta^\delta \times \R} | F_{\ip{ns}+1:\ip{nt}}(x) - F(x)| = o_\Pr(1),
\end{equation}
since $G_{\ip{ns}+1:\ip{nt}}\{F(x)\} = F_{\ip{ns}+1:\ip{nt}}(x)$ for all $(s,t,x) \in \Delta \times \R$. The previous display then implies that
$$
\sup_{(s,t,x) \in \Delta^\delta \times \R} \left| \nu'[F(x) + U^*_{n,s,t,x} \{ F_{\ip{ns}+1:\ip{nt}}(x) - F(x) \}] - \nu'\{F(x)\} \right| = o_\Pr(1)
$$
since $\nu'$ is (uniformly) continuous on $[0,1]$. The latter combined with~\eqref{eq:supdiff2} in which the suprema are restricted to $(s,t) \in \Delta^\delta$ and the fact that $\sup_{(s,t,x) \in \Delta \times \R} | \W_{\alpha,n}\{s,t,F(x)\} | = O_\Pr(1)$ implies finally~\eqref{eq:step1}.

\noindent
{\em Proof of~\eqref{eq:step2}.}
We have 
\begin{multline*}
\sup_{(s,t) \in \Delta} |\X'_{\nu,n}(s,t) - \psi_{\nu,\alpha}(\W_{\alpha,n})(s,t)\} | \leq \sup_{(s,t) \in \Delta} |\X'_{\nu,n}(s,t)| + \sup_{(s,t) \in \Delta} | \psi_{\nu,\alpha}(\W_{\alpha,n})(s,t) |\\ \leq 2 \sup_{u \in [0,1]} |\nu'(u)| \times \sup_{x \in \R} | H_\alpha(x) | \times \sup_{(s,t,u) \in \Delta \times [0,1]} | \W_{\alpha,n}(s,t,u) |.
\end{multline*}
Fix $\eps, \eta > 0$. Proceeding as for the proof of~\eqref{eq:step2}, from the previous inequality, the fact that $\W_{\alpha,n}(s,t,\cdot) = 0$ when $s=t$ and the asymptotic uniform equicontinuity in probability of $\W_{\alpha,n}$, there exists $\delta \in (0,1)$ such that, for all sufficiently large~$n$, 
$$
\Pr \left( \sup_{(s,t) \in \Delta \atop t - s < \delta} |\X'_{\nu,n}(s,t) - \psi_{\nu,\alpha}(\W_{\alpha,n})(s,t)\} | > \eps \right) 
< \eta/2,
$$
and thus, to prove~\eqref{eq:step2}, it remains to show that 
\begin{equation}
\label{eq:ae_X'phi2}
\sup_{(s,t) \in \Delta^\delta} |\X'_{\nu,n}(s,t) - \psi_{\nu,\alpha}(\W_{\alpha,n})(s,t) | = o_\Pr(1),
\end{equation}
where $\Delta^\delta = \{(s,t) \in \Delta : t - s \geq \delta\}$. Proving the previous display is equivalent to showing that $\X'_{\nu,n} - \psi_{\nu,\alpha}(\W_{\alpha,n})$ converges weakly to $(s,t) \mapsto 0$ in $\ell^\infty(\Delta^\delta)$. Fix $(s,t) \in \Delta^\delta$ and let $G(u) = u$ for all $u \in [0,1]$. From~\eqref{eq:unifconv}, $( \W_{\alpha,n}(s,t,\cdot), G_{\ip{ns}+1:\ip{nt}}) \leadsto  ( \W_\alpha(s,t,\cdot), G )$ in $\{ \ell^\infty([0,1]) \}^2$, which implies that
\begin{multline*}
\Big( x \mapsto \nu'\{F(x)\} H_\alpha(x) \W_{\alpha,n}\{s,t,F(x)\}, F_{\ip{ns}+1:\ip{nt}} \Big) \\ \leadsto  \Big(  x \mapsto \nu'\{F(x)\} H_\alpha(x) \W_\alpha\{s,t,F(x)\}, F \Big)
\end{multline*}
in $\{ \ell^\infty(\Rbar) \}^2$, where $\Rbar = [-\infty,\infty]$. The assumptions of the proposition ensure that $x \mapsto \nu'\{F(x)\} H_\alpha(x)$ is a continuous function on $\Rbar$ and therefore that the process 
$$
x \mapsto \nu'\{F(x)\} H_\alpha(x) \W_\alpha\{s,t,F(x)\} \in \ell^\infty(\Rbar)
$$
has continuous trajectories. Combining the latter fact with Lemma~3 of \cite{HolKojQue13} and the continuity of the map $\psi_{\nu,\alpha}$ defined in~\eqref{eq:psi_nu_alpha}, we obtain, from the continuous mapping theorem, that $\X'_{\nu,n}(s,t) - \psi_{\nu,\alpha}(\W_{\alpha,n})(s,t) \leadsto 0$. The latter weak convergence being equivalent to convergence in probability to zero, we have that the finite-dimensional distributions of $\X'_{\nu,n} - \psi_{\nu,\alpha}(\W_{\alpha,n})$ converge weakly to those of $(s,t) \mapsto 0$ in $\ell^\infty(\Delta^\delta)$. 

It remains to show asymptotic tightness. To do so, we shall prove the $\|\cdot\|_1$-asymptotic uniform equicontinuity in probability of $\X'_{\nu,n} - \psi_{\nu,\alpha}(\W_{\alpha,n})$ in $\ell^\infty(\Delta^\delta)$. By the triangle inequality, it suffices to show the latter for $\X'_{\nu,n}$ only as $\psi_{\nu,\alpha}(\W_{\alpha,n})$ converges weakly to $\psi_{\nu,\alpha}(\W_\alpha)$ in $\ell^\infty(\Delta^\delta)$, the limit having continuous trajectories. By Problem~2.1.5 in \cite{vanWel96}, it remains thus to prove that, for any positive sequence $a_n \downarrow 0$,
\begin{equation}
\label{eq:X_nu_n'_auep}
\sup_{(s,t), (s',t') \in \Delta^\delta \atop |s-s'|+|t-t'| \leq a_n} | \X_{\nu,n}'(s,t) - \X_{\nu,n}'(s',t') | = o_\Pr(1).
\end{equation}
The supremum on the left of the previous display is smaller than $A_n + B_n$, where
\begin{multline*}
A_n = \sup_{(s,t), (s',t') \in \Delta^\delta \atop |s-s'|+|t-t'| \leq a_n} \Big| \int_\R  \nu'\{F(x)\}H_\alpha(x) \W_{\alpha,n}\{s,t,F(x)\} \dd F_{\ip{ns}+1:\ip{nt}}(x) \\ - \int_\R  \nu'\{F(x)\} H_\alpha(x) \W_{\alpha,n}\{s',t',F(x) \} \dd F_{\ip{ns}+1:\ip{nt}}(x) \Big|
\end{multline*}
and
\begin{multline*}
B_n = \sup_{(s,t), (s',t') \in \Delta^\delta \atop |s-s'|+|t-t'| \leq a_n} \Big| \int_\R  \nu'\{F(x)\} H_\alpha(x) \W_{\alpha,n}\{s',t',F(x)\} \dd F_{\ip{ns}+1:\ip{nt}}(x) \\ - \int_\R  \nu'\{F(x)\} H_\alpha(x) \W_{\alpha,n}\{s',t',F(x)\} \dd F_{\ip{ns'}+1:\ip{nt'}}(x) \Big|.
\end{multline*}
Concerning $A_n$, we have
$$
A_n \leq \sup_{u \in [0,1]} |\nu'(u)| \times \sup_{x \in \R} | H_\alpha(x) | \times \sup_{(s,t), (s',t') \in \Delta^\delta, u \in [0,1] \atop |s-s'|+|t-t'| \leq a_n} | \W_{\alpha,n}(s,t,u) - \W_{\alpha,n}(s',t',u) | = o_\Pr(1).
$$
The term $B_n$ is smaller than
\begin{multline*}
\sup_{(s,t), (s',t') \in \Delta^\delta \atop |s-s'|+|t-t'| \leq a_n} \left| \frac{1}{\ip{nt}-\ip{ns}} \left\{ \sum_{i=\ip{ns}+1}^{\ip{nt}} \nu'\{F(X_i)\} H_\alpha(X_i) \W_{\alpha,n}\{s',t',F(X_i)\} \right. \right. \\
 \left. \left. - \sum_{i=\ip{ns'}+1}^{\ip{nt'}} \nu'\{F(X_i)\} H_\alpha(X_i) \W_{\alpha,n}\{s',t',F(X_i)\} \right\} \right| + \\ \sup_{(s,t), (s',t') \in \Delta^\delta \atop |s-s'|+|t-t'| \leq a_n} \left| \left( \frac{1}{\ip{nt}-\ip{ns}} - \frac{1}{\ip{nt'}-\ip{ns'}} \right)  \sum_{i=\ip{ns'}+1}^{\ip{nt'}} \nu'\{F(X_i)\} H_\alpha(X_i) \W_{\alpha,n}\{s',t',F(X_i)\} \right|,
\end{multline*}
which is smaller than
\begin{multline*}
2 \times \sup_{(s,t), (s',t') \in \Delta^\delta \atop |s-s'|+|t-t'| \leq a_n} \frac{|\ip{nt} - \ip{nt'}| + |\ip{ns} - \ip{ns'}|}{\ip{nt}-\ip{ns}} \times \sup_{u \in [0,1]} |\nu'(u)| \\ \times \sup_{x \in \R} | H_\alpha(x) | \times \sup_{(s,t,u) \in \Delta \times [0,1]} | \W_{\alpha,n}(s,t,u) | = o_\Pr(1).
\end{multline*}
Hence,~\eqref{eq:X_nu_n'_auep} and thus~\eqref{eq:ae_X'phi2} hold, which implies~\eqref{eq:step2}.

Finally,~\eqref{eq:M_nu_st} is a consequence of the fact that, for any $\nu \in \{\nu_1,\dots,\nu_p\}$, 
$$
\M_{\nu}(s,t) = \Y_{\nu}(s,t) + \int_\R \nu'\{F(x)\} H_\alpha(x) \W_\alpha\{s,t,F(x)\} \dd F(x)
$$
and the analogue properties for $\Y_{\nu}$ and $\W_\alpha$, while~\eqref{eq:covW} follows from standard calculations.
\end{proof}


\section{Proof of Proposition~\ref{prop:weak_S_g_n}}
\label{proof:prop:weak_S_g_n}

Proposition~\ref{prop:weak_S_g_n} is a corollary of a slightly more general result that we shall first prove using Proposition~\ref{prop:weak_Mn} shown in Appendix~\ref{proof:prop:weak_Mn}. To be consistent with Appendix~\ref{proof:prop:weak_Mn}, we keep using the richer notation introduced therein. Let $\nu_1,\dots,\nu_p$ be $p$ functions on $[0,1]$, let $g$ be a function from $\R^p$ to $\R$ defined on $\DD$ containing $\bm \beta_{\bm \nu} = (\beta_{\nu_1},\dots,\beta_{\nu_p})$, and consider the generic statistic
\begin{equation*}
S_{\bm \nu,g,n} =  \max_{1 \leq k \leq n-1} \frac{k (n-k)}{n^{3/2}} \1( \hat{\bm \beta}_{\bm \nu,1:k} \in \DD, \hat{\bm \beta}_{\bm \nu,k+1:n} \in \DD) \left| g( \hat{\bm \beta}_{\bm \nu,1:k}) - g(\hat{\bm \beta}_{\bm \nu,k+1:n}) \right|, 
\end{equation*}
where $\hat{\bm \beta}_{\bm \nu,1:k}$ and $\hat{\bm \beta}_{\bm \nu,k+1:n}$ are defined analogously to~\eqref{eq:hatbbetakl}. When $p=3$, $\nu_1 (x) = 1$, $\nu_2(x) = x$, $\nu_3(x) = x^2$, $x \in [0,1]$, $\DD = \DD_\xi$ in~\eqref{eq:def_xi} and $g \in \{g_\mu,g_\sigma,g_\xi\}$, $S_{\bm \nu,g,n}$ is $S_{g,n}$ in~\eqref{eq:S_g_n}.

\begin{prop}[Limiting null distribution of $S_{\bm \nu,g,n}$]
\label{prop:weak_S_nu_g_n}
If Condition~\ref{cond:main} holds, $\Ex(|X_1|) < \infty$, $\DD$ is an open convex subset of $\R^p$ containing $\bm \beta_{\bm \nu}$ and $g$ is continuously differentiable on~$\DD$, $S_{\bm \nu,g,n}$ converges in distribution to $\sigma_{\bm \nu,g} \sup_{s \in [0,1]} |\U(s) - s \U(1)|$, where $\U$ is a standard Brownian motion on $[0,1]$ and
\begin{equation*}
\sigma_{\bm \nu,g}^2 = \sum_{i,j=1}^p \partial_i g(\bm \beta_{\bm \nu}) \partial_j g(\bm \beta_{\bm \nu}) \Cov(Y_{\nu_i},Y_{\nu_j}),
\end{equation*}
with $Y_{\nu_i}$ defined in~\eqref{eq:Ynu}.
\end{prop}

\begin{proof}
For any $s \in [0,1]$, let
\begin{equation*}
\D_{\bm \nu,g,n}(s) = \sqrt{n} \lambda_n(0,s) \lambda_n(s,1) \1(\hat{\bm \beta}_{\bm \nu,1:\ip{ns}} \in \DD, \hat{\bm \beta}_{\bm \nu,\ip{ns}+1:n} \in \DD) \{ g(\hat{\bm \beta}_{\bm \nu,1:\ip{ns}}) - g(\hat{\bm \beta}_{\bm \nu,\ip{ns}+1:n}) \}, 
\end{equation*}
and notice that 
\begin{equation}
\label{eq:S_nu_g_n_D}
S_{\bm \nu,g,n} = \sup_{s \in [0,1]} |\D_{\bm \nu,g,n}(s)|.
\end{equation}
Furthermore, let 
\begin{equation*}
\M_{\bm \nu,g,n}(s,t) = \sqrt{n} \lambda_n(s,t) \{ g(\hat{\bm \beta}_{\bm \nu,\ip{ns}+1:\ip{nt}}) - g(\bm \beta_{\bm \nu}) \}, \qquad (s,t) \in \Delta.
\end{equation*}
Under $H_0$ in~\eqref{eq:H0} implied by Condition~\ref{cond:main}, it is easy to verify that
\begin{multline}
\label{eq:D_nu_g_n_H0}
\D_{\bm \nu,g,n}(s) = \lambda_n(s,1) \1(\hat{\bm \beta}_{\bm \nu,1:\ip{ns}} \in \DD, \hat{\bm \beta}_{\bm \nu,\ip{ns}+1:n} \in \DD) \M_{\bm \nu,g,n}(0,s) \\ - \lambda_n(0,s) \1(\hat{\bm \beta}_{\bm \nu,1:\ip{ns}} \in \DD, \hat{\bm \beta}_{\bm \nu,\ip{ns}+1:n} \in \DD)  \M_{\bm \nu,g,n}(s,1), \qquad s \in [0,1]. 
\end{multline}
To show the desired result, we shall first prove that
\begin{equation}
\label{eq:ae1}
A_n = \sup_{s \in [0,1]} \1(\hat{\bm \beta}_{\bm \nu,1:\ip{ns}} \in \DD, \hat{\bm \beta}_{\bm \nu,\ip{ns}+1:n} \in \DD) \left| \M_{\bm \nu,g,n}(0,s) - \sum_{i=1}^p \partial_i g(\bm \beta_{\bm \nu})  \M_{\nu_i,n}(0,s)\right| = o_\Pr(1)
\end{equation}
and that
\begin{equation}
\label{eq:ae2}
B_n = \sup_{s \in [0,1]} \1(\hat{\bm \beta}_{\bm \nu,1:\ip{ns}} \in \DD, \hat{\bm \beta}_{\bm \nu,\ip{ns}+1:n} \in \DD) \left| \M_{\bm \nu,g,n}(s,1) - \sum_{i=1}^p \partial_i g(\bm \beta_{\bm \nu})  \M_{\nu_i,n}(s,1)\right| = o_\Pr(1),
\end{equation}
where $\M_{\nu_1,n},\dots,\M_{\nu_p,n}$ are defined in~\eqref{eq:M_nu_n}.

Let us start with~\eqref{eq:ae1}. Using the mean value theorem when the indicators are equal to one, $A_n$ is smaller than $\sum_{i=1}^p A_{i,n}$, where
\begin{equation}
\label{eq:Ain}
A_{i,n} = \sup_{s \in [0,1]}  \1(\hat{\bm \beta}_{\bm \nu,1:\ip{ns}} \in \DD) | \partial_i g(\hat{\bm \beta}_{\bm \nu,1:\ip{ns}}^*) - \partial_i g(\bm \beta_{\bm \nu}) | |\M_{\nu_i,n}(0,s) |,
\end{equation}
and where $\hat{\bm \beta}_{\bm \nu,1:\ip{ns}}^*$ lies on the segment between $\bm \beta_{\bm \nu}$ and $\hat{\bm \beta}_{\bm \nu,1:\ip{ns}}$. To show~\eqref{eq:ae1}, it suffices therefore to prove that $A_{i,n} = o_\Pr(1)$ for all $i \in \{1,\dots,p\}$.

Fix $\nu \in \{\nu_1,\dots,\nu_p\}$. Since $\beta_\nu < \infty$, from the strong law of large number, we have that $\tilde \beta_{\nu,1:n} \as \beta_\nu$, where $\tilde \beta_{\nu,1:n}$ is defined analogously to~\eqref{eq:tilde_beta_st}. Thus,
\begin{equation*}
\left| \tilde \beta_{\nu,1:n} - \hat \beta_{\nu,1:n} \right| \leq \sup_{x \in \R} | \nu\{F_{1:n}(x)\} - \nu\{F(x)\} | \times \frac{1}{n} \sum_{i=1}^n  | X_i | \as 0
\end{equation*}
since $\Ex(|X_1|) < \infty$, and the supremum on the right converges to 0 almost surely as a consequence of the Glivenko-Cantelli lemma and the continuous mapping theorem. It follows that $\hat \beta_{\nu,1:n} \as \beta_\nu$ for all $\nu \in \{\nu_1,\dots,\nu_p\}$, and therefore that $\hat{\bm \beta}_{\bm \nu,1:n} \as \bm \beta_{\bm \nu}$. 

Fix $i \in \{1,\dots,p\}$. The continuous mapping theorem next implies that $\partial_i g(\hat{\bm \beta}_{\bm \nu,1:n}) \as \partial_i g( \bm \beta_{\bm \nu} )$, which in turn implies that $\sup_{k \geq 1} \1(\hat{\bm \beta}_{\bm \nu,1:k} \in \DD) | \partial_i g(\hat{\bm \beta}_{\bm \nu,1:k}) - \partial_i g( \bm \beta_{\bm \nu} ) | < \infty$ almost surely. 

Let $\eta > 0$. Then, there exists $M > 0$ such that
\begin{equation}
\label{eq:bound1}
\Pr \left\{ \sup_{k \geq 1} \1(\hat{\bm \beta}_{\bm \nu,1:k} \in \DD) \left| \partial_i g(\hat{\bm \beta}_{\bm \nu,1:k}) - \partial_i g( \bm \beta_{\bm \nu} ) \right| > M \right\} \leq \eta/4.
\end{equation}
Furthermore, let $\eps > 0$ and consider $\eps/M > 0$. Using the asymptotic uniform equicontinuity in probability of $\M_{\nu_i,n}$ (a consequence of Proposition~\ref{prop:weak_Mn}), there exists $\delta \in (0,1)$ such that, for all sufficiently large $n$, 
\begin{equation}
\label{eq:bound2}
\Pr \left\{ \sup_{s \in [0,\delta)} |\M_{\nu_i,n}(0,s)| > \eps/M \right\} \leq \eta/4.
\end{equation}
Then,
\begin{multline}
\label{eq:prob_ineq}
\Pr(A_{i,n} > \eps) \leq \Pr\left\{ \sup_{s \in [0,\delta)} \1(\hat{\bm \beta}_{\bm \nu,1:\ip{ns}} \in \DD) | \partial_i g(\hat{\bm \beta}_{\bm \nu,1:\ip{ns}}^*) - \partial_i g(\bm \beta_{\bm \nu}) | |\M_{\nu_i,n}(0,s) |  > \eps \right\} \\ + \Pr\left\{ \sup_{s \in [\delta,1]} \1(\hat{\bm \beta}_{\bm \nu,1:\ip{ns}} \in \DD) | \partial_i g(\hat{\bm \beta}_{\bm \nu,1:\ip{ns}}^*) - \partial_i g(\bm \beta_{\bm \nu}) | |\M_{\nu_i,n}(0,s) | \} > \eps \right\}.
\end{multline}
The first term on the right is smaller, for all sufficiently large $n$, than  
$$
\Pr \left\{ \sup_{k \geq 1} \1(\hat{\bm \beta}_{\bm \nu,1:k} \in \DD) \left| \partial_i g(\hat{\bm \beta}_{\bm \nu,1:k}) - \partial_i g( \bm \beta_{\bm \nu} ) \right| > M \right\} + \Pr \left\{ \sup_{s \in [0,\delta)} |\M_{\nu_i,n}(0,s)| > \eps/M \right\} \leq \eta/2,
$$
where the last inequality follows by~\eqref{eq:bound1} and~\eqref{eq:bound2}. The second term on the right of~\eqref{eq:prob_ineq} is smaller, for all sufficiently large $n$, than $\eta/2$. Indeed,
\begin{multline*}
\sup_{s \in [\delta,1]} \1(\hat{\bm \beta}_{\bm \nu,1:\ip{ns}} \in \DD) | \partial_i g(\hat{\bm \beta}_{\bm \nu,1:\ip{ns}}^*) - \partial_i g(\bm \beta_{\bm \nu}) | |\M_{\nu_i,n}(0,s) | \\ \leq \sup_{s \in [\delta,1]} | \partial_i g(\hat{\bm \beta}_{\bm \nu,1:\ip{ns}}^*) - \partial_i g(\bm \beta_{\bm \nu}) |  \times \sup_{s \in [0,1]} |\M_{\nu_i,n}(0,s) | = o_\Pr(1)
\end{multline*}
since the second supremum is $O_\Pr(1)$ from the weak convergence of $\M_{\nu_i,n}$, while the first supremum is $o_\Pr(1)$ as a consequence of the fact that, for any $i \in \{1,\dots,p\}$, 
$$
\sup_{s \in [\delta,1]} | \hat \beta_{\nu_i,1:\ip{ns}} - \beta_{\nu_i} | = n^{-1/2} \sup_{s \in [\delta,1]} \lambda_n(0,s)^{-1} \times \sup_{s \in [0,1]} | \M_{\nu_i,n}(0,s) | =  o_\Pr(1)
$$
and the continuous mapping theorem. Hence, $A_{i,n}$ defined in~\eqref{eq:Ain} converges in probability to zero for all $i \in \{1,\dots,p\}$, which implies that~\eqref{eq:ae1} holds. The fact that~\eqref{eq:ae2} holds follows immediately since $B_n$ in~\eqref{eq:ae2} (written as a maximum) is nothing else than $A_n$ in~\eqref{eq:ae1} computed from the sequence $X_n,\dots,X_1$. Therefore, for any $\eps > 0$, $\Pr (B_n > \eps) = \Pr (A_n > \eps) \to 0$ by~\eqref{eq:ae1}.

Next, let 
$$
S_{\bm \nu,g,n}' =  \sup_{s \in [0,1]} \1(\hat{\bm \beta}_{\bm \nu,1:\ip{ns}} \in \DD, \hat{\bm \beta}_{\bm \nu,\ip{ns}+1:n} \in \DD) | \D_{\bm \nu,g,n}'(s) |,
$$
where 
\begin{equation}
\label{eq:D'_nu_g_n}
\D_{\bm \nu,g,n}'(s) = \sum_{i=1}^p \partial_i g(\bm \beta_{\bm \nu}) \{ \lambda_n(s,1)  \M_{\nu_i,n}(0,s) - \lambda_n(0,s) \M_{\nu_i,n}(s,1) \}, \qquad s \in [0,1].
\end{equation}
Then, starting from~\eqref{eq:S_nu_g_n_D} and~\eqref{eq:D_nu_g_n_H0}, and using the reverse triangle inequality, we obtain that
\begin{equation}
\label{eq:conv0}
\left| S_{\bm \nu,g,n} - S_{\bm \nu,g,n}' \right| \leq A_n + B_n = o_\Pr(1),
\end{equation}
where $A_n$ and $B_n$ are defined in~\eqref{eq:ae1} and~\eqref{eq:ae2}. 

We shall now verify that 
\begin{equation}
\label{eq:conv}
S_{\bm \nu,g,n}' - S_{\bm \nu,g,n}''' = o_\Pr(1), 
\end{equation}
where
\begin{equation}
\label{eq:S'''_nu_g_n} 
S_{\bm \nu,g,n}''' =  \sup_{s \in [0,1]}  | \D_{\bm \nu,g,n}'(s) |. 
\end{equation}
To do so, we use the decomposition $S_{\bm \nu,g,n}' - S_{\bm \nu,g,n}''' = S_{\bm \nu,g,n}' - S_{\bm \nu,g,n}'' + S_{\bm \nu,g,n}'' - S_{\bm \nu,g,n}'''$, where
$$
S_{\bm \nu,g,n}'' =  \sup_{s \in [0,1]}  \1(\hat{\bm \beta}_{\bm \nu,\ip{ns}+1:n} \in \DD) | \D_{\bm \nu,g,n}'(s) |.
$$
Since $\hat{\bm \beta}_{\bm \nu,1:k} \as \bm \beta_{\bm \nu}$ as $k \to \infty$, we have that, for almost every $\omega$, there exists $k_\omega$ such that $k \geq k_\omega$ implies that $\hat{\bm \beta}_{\bm \nu,1:k}(\omega) \in \DD$. Hence, for almost every $\omega$, 
\begin{multline*}
S_{\bm \nu,g,n}'(\omega) - S_{\bm \nu,g,n}''(\omega) = \sup_{s \in [k_\omega/n,1]} \1\{\hat{\bm \beta}_{\bm \nu,\ip{ns}+1:n}(\omega) \in \DD\} | \D_{\bm \nu,g,n}'(s,\omega) |  \\  - \sup_{s \in [0,1]} \1\{ \hat{\bm \beta}_{\bm \nu,\ip{ns}+1:n}(\omega) \in \DD \} | \D_{\bm \nu,g,n}'(s,\omega) |  \to 0,
\end{multline*}
which implies that $S_{\bm \nu,g,n}' - S_{\bm \nu,g,n}'' = o_\Pr(1)$. Similarly, we obtain that
$$
\sup_{s \in [0,1]} \1(\hat{\bm \beta}_{\bm \nu,1:\ip{ns}} \in \DD) | \D_{\bm \nu,g,n}'(s) | - S_{\bm \nu,g,n}''' = o_\Pr(1).
$$
Then, using the fact that $S_{\bm \nu,g,n}'' - S_{\bm \nu,g,n}'''$ is $\sup_{s \in [0,1]} \1(\hat{\bm \beta}_{\bm \nu,1:\ip{ns}} \in \DD) | \D_{\bm \nu,g,n}'(s) |  - S_{\bm \nu,g,n}'''$ computed from the sequence $X_n,\dots,X_1$, we obtain that, for every $\eps > 0$,
$$
\Pr( | S_{\bm \nu,g,n}'' - S_{\bm \nu,g,n}''' | > \eps) = \Pr \left\{ \left| \sup_{s \in [0,1]} \1(\hat{\bm \beta}_{\bm \nu,1:\ip{ns}} \in \DD) | \D_{\bm \nu,g,n}'(s) |  - S_{\bm \nu,g,n}''' \right| > \eps \right\} \to 0.
$$
Therefore,~\eqref{eq:conv} holds. Next, from Proposition~\ref{prop:weak_Mn} and the continuous mapping theorem, we have that $\D_{\bm \nu,g,n}' \leadsto \D_{\bm \nu,g}'$ in $\ell^\infty([0,1])$, where $\D_{\bm \nu,g,n}'$ is defined in~\eqref{eq:D'_nu_g_n},
$$
\D_{\bm \nu,g}'(s) = \sum_{i=1}^p \partial_i g(\bm \beta_{\bm \nu}) \{ (1-s)  \M_{\nu_i}(0,s) - s \M_{\nu_i}(s,1) \} = \sum_{i=1}^p \partial_i g(\bm \beta_{\bm \nu}) \{ \M_{\nu_i}(0,s) - s \M_{\nu_i}(0,1) \}, 
$$
for $s \in [0,1]$, and $\M_{\nu_1}, \dots, \M_{\nu_p}$ are defined in Proposition~\ref{prop:weak_Mn}. The desired convergence in distribution of $S_{\bm \nu,g,n}$ finally follows from~\eqref{eq:conv0},~\eqref{eq:conv},~\eqref{eq:S'''_nu_g_n} and standard calculations.
\end{proof}

\begin{proof}[\bf Proof of Proposition~\ref{prop:weak_S_g_n}]
The convergence in distribution of $S_{g,n}$ is an immediate consequence of Proposition~\ref{prop:weak_S_nu_g_n}. It remains to prove that $S_{g,n} - T_{g,n} =o_\Pr(1)$. The latter is a consequence of the fact that, for any $i \in \{1,2,3\}$,
\begin{equation}
\label{eq:aim1}
I_n = \sup_{s \in [0,1]} | \M_{\nu_i,n}(0,s) - \sqrt{n} \lambda_n(0,s) ( \hat b_{i,1:\ip{ns}} - \beta_{\nu_i} ) | = o_\Pr(1)
\end{equation}
and
\begin{equation}
\label{eq:aim2}
J_n = \sup_{s \in [0,1]} | \M_{\nu_i,n}(s,1) - \sqrt{n} \lambda_n(s,1) ( \hat b_{i,\ip{ns}+1:n} - \beta_{\nu_i} ) | = o_\Pr(1),
\end{equation}
where $\hat b_{i,1:\ip{ns}}$ and $\hat b_{i,\ip{ns}+1:n}$ are defined analogously to~\eqref{eq:landwehr}. The above statements are clearly true for $i=1$ as $\hat b_{1,k:l} = \hat \beta_{\nu_1,k:l}$ for all integers $1 \leq k, l \leq n$. For $i=2$, it can be verified that the supremum on the left of~\eqref{eq:aim1} restricted to $s \in [2/n,1]$ can be rewritten as
\begin{multline*}
\sup_{s \in [2/n,1]} \frac{1}{\sqrt{n}} \left| \sum_{i=1}^{\ip{ns}} X_i F_{1:\ip{ns}}(X_i)- \sum_{i=1}^{\ip{ns}} X_i \frac{\ip{ns} F_{1:\ip{ns}}(X_i) - 1}{\ip{ns} - 1} \right| \\
\leq \frac{2}{\sqrt{n}} \sup_{s \in [2/n,1]} \frac{1}{\ip{ns} - 1} \sum_{i=1}^{\ip{ns}} |X_i| 
\leq \frac{4}{\sqrt{n}} \sup_{k \geq 1} \frac{1}{k} \sum_{i=1}^{k} |X_i| \as 0
\end{multline*}
since $\Ex(|X_1|) < \infty$. The case $i=3$ is similar. To show~\eqref{eq:aim2}, we use the fact that $J_n$ is $I_n$ computed from the sequence $X_n,\dots,X_1$ and conclude from~\eqref{eq:aim1}. 

Starting from~\eqref{eq:aim1} and~\eqref{eq:aim2}, it is possible to adapt the proof of Proposition~\ref{prop:weak_S_nu_g_n} to show that $S_{g,n} - T_{g,n} = o_\Pr(1)$. We omit the details for the sake of brevity.  
\end{proof}


\bibliographystyle{plainnat}
\bibliography{biblio}

\begin{thebibliography}{24}
\providecommand{\natexlab}[1]{#1}
\providecommand{\url}[1]{\texttt{#1}}
\expandafter\ifx\csname urlstyle\endcsname\relax
  \providecommand{\doi}[1]{doi: #1}\else
  \providecommand{\doi}{doi: \begingroup \urlstyle{rm}\Url}\fi

\bibitem[Bernard et~al.(2013)Bernard, Naveau, Vrac, and Mestre]{ClusterMax}
E.~Bernard, P.~Naveau, M.~Vrac, and O.~Mestre.
\newblock Clustering of maxima: {S}patial dependencies among heavy rainfall in
  {F}rance.
\newblock \emph{Journal of Climate}, 26\penalty0 (2):\penalty0 7929--7937,
  2013.
\newblock URL
  \url{http://cse.ipsl.fr/outils-logiciels/90-pakage-r-clustermax-clustering-pour-valeurs-extremes}.

\bibitem[B\"ucher and Kojadinovic(2016)]{BucKoj16b}
A.~B\"ucher and I.~Kojadinovic.
\newblock Dependent multiplier bootstraps for non-degenerate $u$-statistics
  under mixing conditions with applications.
\newblock \emph{Journal of Statistical Planning and Inference}, 170:\penalty0
  83--105, 2016.

\bibitem[Cs\"org\H{o} and Horv{\'a}th(1997)]{CsoHor97}
M.~Cs\"org\H{o} and L.~Horv{\'a}th.
\newblock \emph{Limit theorems in change-point analysis}.
\newblock Wiley Series in Probability and Statistics. John Wiley \& Sons,
  Chichester, UK, 1997.

\bibitem[de~Haan and Ferreira(2006)]{DehFer06}
L.~de~Haan and A.~Ferreira.
\newblock \emph{Extreme value theory: {A}n introduction}.
\newblock Springer, 2006.

\bibitem[Diebolt et~al.(2008)Diebolt, Guillou, Naveau, and
  Ribereau]{DieGuiNavRib08}
J.~Diebolt, A.~Guillou, P.~Naveau, and P.~Ribereau.
\newblock Improving probability-weighted moment methods for the generalized
  extreme-value distribution.
\newblock \emph{REVSTAT}, 6\penalty0 (1):\penalty0 33--50, 2008.

\bibitem[Ferreira and {de Haan}(2015)]{FerdeH15}
A.~Ferreira and L.~{de Haan}.
\newblock On the block maxima method in extreme value theory: {PWM} estimators.
\newblock \emph{The Annals of Statistics}, 43\penalty0 (1):\penalty0 276--298,
  2015.

\bibitem[Genest and Segers(2009)]{GenSeg09}
C.~Genest and J.~Segers.
\newblock Rank-based inference for bivariate extreme-value copulas.
\newblock \emph{The Annals of Statistics}, 37:\penalty0 2990--3022, 2009.

\bibitem[Gilleland and Katz(2011)]{extRemes}
E.~Gilleland and R.W. Katz.
\newblock New software to analyze how extremes change over time.
\newblock \emph{Eos}, 92\penalty0 (2):\penalty0 13--14, 2011.

\bibitem[Greenwood et~al.(1979)Greenwood, Landwehr, and Wallis]{GreLanWal79}
J.A. Greenwood, J.M. Landwehr, and J.R. Wallis.
\newblock Probability-weighted moments: {D}efinition and relation to parameters
  of several distributions expressable in inverse form.
\newblock \emph{Water Resources Research}, 15:\penalty0 1049--1054, 1979.

\bibitem[Guillou et~al.(2009)Guillou, Naveau, Diebolt, and
  Ribereau]{GuiNavDieRib09}
A.~Guillou, P.~Naveau, J.~Diebolt, and P.~Ribereau.
\newblock Return level bounds for discrete and continuous random variables.
\newblock \emph{Test}, 18:\penalty0 584--604, 2009.

\bibitem[Gumbel(1958)]{Gum58}
E.~J. Gumbel.
\newblock \emph{Statistics of extremes}.
\newblock Columbia University Press, New York, 1958.

\bibitem[Heffernan and Stephenson.(2014)]{ismev}
J.E. Heffernan and A.G. Stephenson.
\newblock \emph{ismev: {A}n Introduction to Statistical Modeling of Extreme
  Values}, 2014.
\newblock URL \url{http://CRAN.R-project.org/package=ismev}.
\newblock R package version 1.40.

\bibitem[Holmes et~al.(2013)Holmes, Kojadinovic, and Quessy]{HolKojQue13}
M.~Holmes, I.~Kojadinovic, and J-F. Quessy.
\newblock Nonparametric tests for change-point detection \`a la {G}ombay and
  {H}orv\'ath.
\newblock \emph{Journal of Multivariate Analysis}, 115:\penalty0 16--32, 2013.

\bibitem[Hosking and Wallis(1987)]{HosWal87}
J.R.M. Hosking and J.R. Wallis.
\newblock Parameter and quantile estimation for the generalized {P}areto
  distribution.
\newblock \emph{Technometrics}, 29:\penalty0 339--349, 1987.

\bibitem[Hosking et~al.(1985)Hosking, Wallis, and Wood]{HosWalWoo85}
J.R.M. Hosking, J.R. Wallis, and E.F. Wood.
\newblock Estimation of the generalized extreme-value distribution by the
  method of probability-weighted moments.
\newblock \emph{Technometrics}, 27:\penalty0 251--261, 1985.

\bibitem[{Jaru\v skova} and {Rencov\'a}(2008)]{JarRen08}
D.~{Jaru\v skova} and M.~{Rencov\'a}.
\newblock {Analysis of annual maximal and minimal temperatures for some
  European cities by change point methods}.
\newblock \emph{Environmetrics}, 19:\penalty0 221--233, 2008.

\bibitem[Kojadinovic(2015)]{npcp}
I.~Kojadinovic.
\newblock \emph{npcp: {S}ome Nonparametric Tests for Change-Point Detection in
  Possibly Multivariate Observations}, 2015.
\newblock URL \url{http://CRAN.R-project.org/package=npcp}.
\newblock R package version 0.1-5.

\bibitem[Kosorok(2008)]{Kos08}
M.R. Kosorok.
\newblock \emph{Introduction to empirical processes and semiparametric
  inference}.
\newblock Springer, New York, 2008.

\bibitem[Landwehr et~al.(1979)Landwehr, Matalas, and Wallis]{LanMatWal79}
J.M. Landwehr, N.C. Matalas, and J.R. Wallis.
\newblock Probability-weighted moments compared with some traditional
  techniques in estimating {G}umbel parameters and quantiles.
\newblock \emph{Water Resources Research}, 15:\penalty0 1055--1064, 1979.

\bibitem[Leadbetter et~al.(1983)Leadbetter, Lindgren, and
  Rootz\'en]{LeaLinRoo83}
M.R. Leadbetter, G.~Lindgren, and H.~Rootz\'en.
\newblock \emph{Extremes and related properties of random sequences and
  processes}.
\newblock Springer Series in Statistics. Springer-Verlag, New York, 1983.

\bibitem[Phillips(1987)]{Phi87}
P.C.B. Phillips.
\newblock Time series regression with unit roots.
\newblock \emph{Econometrica}, 55:\penalty0 277--301, 1987.

\bibitem[{R Core Team}(2016)]{Rsystem}
{R Core Team}.
\newblock \emph{{R}: {A} Language and Environment for Statistical Computing}.
\newblock R Foundation for Statistical Computing, Vienna, Austria, 2016.
\newblock URL \url{http://www.R-project.org}.

\bibitem[Stephenson(2002)]{evd}
A.G. Stephenson.
\newblock {evd: Extreme Value Distributions}.
\newblock \emph{R News}, 2\penalty0 (2):\penalty0 0, June 2002.
\newblock URL \url{http://CRAN.R-project.org/doc/Rnews/}.

\bibitem[{van der Vaart} and Wellner(2000)]{vanWel96}
A.W. {van der Vaart} and J.A. Wellner.
\newblock \emph{Weak convergence and empirical processes}.
\newblock Springer, New York, 2000.
\newblock Second edition.

\end{thebibliography}


\end{document}